\documentclass[12pt,a4paper]{article}  
\usepackage{subfig}
\usepackage{array}
\usepackage[vlined,ruled]{algorithm2e}
\pdfoutput=1
\usepackage{adjustbox}
\usepackage{lipsum}
\usepackage{natbib}
\usepackage{lipsum}                          
\usepackage[utf8]{inputenc}                  
\usepackage{a4wide}                          
\usepackage{array}                           
\usepackage{dsfont}                          
\usepackage{amsmath}                         
\usepackage{amsfonts}                        
\usepackage{amssymb}                         
\usepackage{mathrsfs}                        
\usepackage{graphicx}                
\usepackage{sidecap}                         
\usepackage{enumerate}                       
\usepackage{floatrow}                        
\usepackage{bm}                              
\usepackage[framemethod=default]{mdframed}   
\usepackage[dvipsnames]{xcolor}              
\usepackage{fontawesome}                     
\usepackage{sidecap}                         
\usepackage{tabulary}                        
\usepackage[sf,sl,outermarks]{titlesec}      
\usepackage[subfigure]{tocloft}              
\usepackage{booktabs}                        
\usepackage{siunitx}                         
\usepackage{chngcntr}                        
\usepackage[
  tableposition = top,
  labelsep      = period,
  justification = raggedright,
  format        = hang,
  ]{caption}                                 
\usepackage[
  colorlinks  = true,
  linkcolor   = {blue!50!black},
  urlcolor    = blue,
  citecolor   = {blue!50!black},
  linktocpage = true,
  ]{hyperref}

\usepackage{amsthm}
\usepackage[bottom]{footmisc}

\usepackage[autostyle, english = american]{csquotes} 
\MakeOuterQuote{"}
\renewcommand{\thesection}{\arabic{section}}         
\renewcommand{\thesubsection}{\Alph{subsection}}    
\advance\cftsecnumwidth 0.2em\relax                 
\advance\cftsubsecindent 0.2em\relax                
\advance\cftsubsecnumwidth 0.2em\relax              
\counterwithout*{subsection}{section}               

\usepackage{accents}

\titleformat{\section}[hang]  
  {\large\bfseries}      
  {\thesection.}              
  {2ex}                       
  {\centering}  

\titleformat{\subsection}[hang] 
  {\normalfont\bfseries}        
  {\thesubsection.}             
  {1ex}                         
  {\centering}                  

\titleformat{\subsubsection}[hang]  
  {\normalfont\bfseries}            
  {}                                
  {1ex}                             
  {\centering}                      

\newtheorem{assump1}{Assumption}

\newtheorem{proposition}{Proposition}

\newtheoremstyle{mystyle}
  {}
  {}
  {\itshape}
  {}
  {\bfseries}
  {}
  { }
  {}
\let\oldproofname=\proofname
\renewcommand{\proofname}{\rm\bf{\oldproofname}}
\theoremstyle{definition}

\numberwithin{equation}{section}

\usepackage{stackengine}


\graphicspath{{./pic/}}

\definecolor{gray}{gray}{0.95}
\newmdenv[linecolor=white,backgroundcolor=gray]{grayframe}

\title{Deep Partial Least Squares for Instrumental Variable Regression}
\author{	
	\makebox[.4\linewidth]{Maria Nareklishvili}\\
	\textit{\small  Booth School of Business}\\
	\textit{\small  University of Chicago}\\
    \and
	\makebox[.4\linewidth]{Nicholas Polson\footnote{Email: ngp@chicagobooth.edu}}\\
	\textit{\small  Booth School of Business}\\
	\textit{\small  University of Chicago}\\
	\and 
	\makebox[.4\linewidth]{Vadim Sokolov}\\
	\textit{\small  Department of Systems Engineering }\\
	\textit{\small  and Operations Research}\\
	\textit{\small  George Mason University}\\
}

\begin{document}
\maketitle

\begin{abstract} \noindent 
In this paper, we propose deep partial least squares for the estimation of high-dimensional nonlinear instrumental variable regression. As a precursor to a flexible deep neural network architecture, our methodology uses partial least squares for dimension reduction and feature selection from the set of instruments and covariates. A central theoretical result, due to \cite{brillinger2012generalized}, shows that the feature selection provided by partial least squares is consistent and the weights are estimated up to a proportionality constant. We illustrate our methodology with synthetic datasets with a sparse and correlated network structure and draw applications to the effect of childbearing on the mother's labor supply based on classic data of \cite{angrist1996children}.  The results on synthetic data as well as applications show that the deep partial least squares method significantly outperforms other related methods.  Finally, we conclude with directions for future research.
\end{abstract}


\section{Introduction}

Nonlinear instrumental variable (IV) regression is a vital tool for estimating the causal effect of exposure on certain outcomes. In fact, virtually all legitimate techniques for causal inference can be seen as manifestations of instrumental variables, including but not limited to randomized clinical trials (considered a perfect instrument), intention-to-treat analysis (stemming from random incentive allocation), natural experiments, and regression discontinuity (\citealp{heckman1995randomization}).  A valid instrumental variable meets the following assumptions: it has a significant effect on the treatment (i.e., relevance). It does not influence the outcome directly, through channels other than the treatment (i.e., the exclusion restriction), and a valid instrument is not associated with the unobserved characteristics that affect the outcome (i.e., exogeneity). Under these assumptions, the IV approach recovers consistent coefficients of interest. Traditionally, the techniques to estimate the causal parameter in the IV approach are well-suited for low and middle-scale data (\citealp{white1982instrumental, terza2008two}). These techniques most often fail with high-dimensional instruments (\citealp{mccullagh2018statistical}) and struggle to extract meaningful insights when instruments sparsely relate to each other.  Several papers investigate the approach with many, possibly correlated and sparse instrumental variables \citep{belloni2011inference, chernozhukov2015valid, chernozhukov2015post}. 

For predictive tasks, contemporary research has predominantly centered on narrow deep neural networks \citep{hartford2016counterfactual, polson2017deep,  liu2020deep, beise2021decision}, which are distinguished by their "self-featurizing" basis functions. In other words, feature extraction and dimension reduction are integrated into the pattern-matching algorithm. However, there are limitations to this approach, including a lack of theoretical understanding, difficulty in quantifying uncertainty, and limited capacity for probabilistic reasoning. To enhance the efficiency of deep neural networks, some authors introduce linear methods as a precursor to implementing a deep learning model.  \cite{debiolles2004combined, jia2016optimized} use partial least squares (PLS) to extract features before learning the output with deep neural networks. PLS is a linear method that selects relevant dimensions that are particularly relevant for predicting the outcome. 

To mitigate the issues inherent to high-dimensional instrumental variables, we propose deep partial least squares for IV regression. One can view our approach as merging two cultures of machine learning tools and statistical inference (\citealp{brillinger2012generalized, sarstedt2022progress, monis2022efficient}). Specifically,  we employ linear techniques, such as partial least squares to identify meaningful instruments and form predictors that are ensemble averages of deep learners. We show that, under several data-generating process assumptions, the method exhibits desirable large sample properties.

A traditional approach to estimating the IV regression is a  two-stage least-squares (2SLS) technique. In the first stage, valid instrumental variables predict the treatment. In the second stage, we estimate the effect of the predicted treatment (and possibly other control variables) on the outcome of interest (\citealp{mogstad2021causal}). When the relation of the treatment and instruments is nonlinear, \cite{terza2008two} show that the 2SLS method is inconsistent. \cite{iwata2001recentered} uses a re-centered and re-scaled outcome variable and proves the consistency of the method for censored and truncated data. \cite{terza2008two} and \cite{adkins2012testing} use predicted treatment residuals in addition to the treatment and other exogenous variables in the second-stage regression. Controlling for residuals in the outcome regression has a long history (\citealp{ florens2002instrumental, heckman2004using, navarro2010control}). The methods introduced by \cite{terza2008two, adkins2012testing} can also be identified as control function methods.

Our goal in this article is to extend the deep partial least squares (DPLS) method to extract relevant instruments in the first stage. More importantly, when the observed outcome variables consist of errors, we show that the method is consistent under a re-centered and re-scaled outcome described by \cite{iwata2001recentered} or after controlling for the predicted residuals in the second stage regression as illustrated by \cite{heckman2004using}.  The first layer of DPLS consists of a partial least squares method to extract features via hyperplanes in a high-dimensional setting. The subsequent layers efficiently post-process them based on ReLU networks as a deep learner (\citealp{yarotsky2017error}). PLS is an attractive method for feature extraction. However, as an alternative, the generalized method of moments (GMM, \citealp{hansen1982generalized}) can also be applied to our feature selection stage to provide variance stabilized estimators. This results in a statistical improvement on traditional stochastic descent estimators that are commonplace in machine learning (see, for example, \citealp{iwata2017improving}).

Several related articles apply instrumental variables to address challenges and research questions in business and industry. \cite{abadie2002instrumental} examine the effect of job training program participation on earnings.  By comparison, \cite{chernozhukov2004effects} investigate the effect of 401(k) participation on wealth.  \cite{moreno2014doing} mine text sentiment from user-generated comments on an online service platform, and subsequently estimate the impact of
(predicted) sentiment on buyers’ purchasing decisions. \cite{hartford2017deep} estimate the effect of the airline ticket price on demand. Another related article \citep{cengiz2022seeing} builds a boosted tree model to identify minimum wage workers based on their demographics, and then examines the effect of minimum wage policies on labor market outcomes for these workers. In this article, we revisit the effect of childbearing on women's labor supply \citep{angrist1996children}. The instruments consist of the second and third kids being twins and the interactions with parental characteristics that are relevant for predicting the mother's labor supply. 

In addition, we conduct two simulation experiments to demonstrate the predictive performance of DPLS in IV regression. The aim of these experiments is to mimic high-dimensional and sparse data, most often prevalent in business, finance, and economics. In the first experiment, we consider a high-dimensional IV space, where some of the instruments are redundant for predicting the treatment. The instruments are uncorrelated with each other. In the second experiment, we extend the first one by introducing a network structure among the instruments, as proposed by \cite{jeong2003measuring}. In this setup, the instruments are sparsely correlated, and some of them are not related to the treatment. Our aim is to illustrate whether DPLS can identify relevant instruments and improve prediction accuracy in both scenarios.

Statistical properties of deep neural networks are sparse but growing. \cite{polson2018posterior} discuss theoretical foundations of sparse rectified-linear-unit (ReLU) networks and the advantage of using Spike-and-Slab prior as an alternative to Dropout. They show that the resulting posterior prediction of ReLU networks with Spike-and-Slab regularization converge to a true function at a rate of $log^{\delta}(n)/n^{-K}$ for $\delta > 1$ and a positive constant $K$ (with $n$ number of observations). \cite{polson2019bayesian} provide a theoretical connection between the Spike-and-Slab priors and $L_0$ norm regularization. They demonstrate that the regularized estimators can result in improved out-of-sample prediction performance. To emphasize the advantages of regularization in deep neural networks,  \cite{bhadra2020horseshoe} discuss theoretical and empirical justifications (as well as challenges) for the horseshoe prior. 

By comparison, \cite{hanin2019deep} derive the sharp upper bound for the number of activation function regions in ReLU neural networks. They find that this number in practice is far from the maximum possible and depends on the number of neurons in the network, rather than the depth.  A recent paper by \cite{farrell2021deep} provides groundbreaking results on the asymptotic theory of deep neural networks. Under the assumption that the number of hidden layers, i.e., the depth of the network, grows with the sample size, they provide a high probability convergence rate for ReLU neural networks.

Our paper sheds light on the dimension reduction importance in sparse ReLU networks. Specifically, DPLS provides a simple, interpretable framework for modeling instrumental variables when the policy (treatment) and outcome are measured with errors. DPLS consists of a system of equations that can theoretically be viewed as an infinite sequential generalization of 2SLS. Beyond the consistency of this method, we demonstrate that the shrinkage in the first layer can substantially improve prediction performance in an instrumental variable regression with many instruments.

The results based on simulated experiments and the application show that DPLS significantly outperforms other related methods. Specifically, the first stage prediction performance improvement of DPLS relative to OLS, LASSO (\citealp{tibshirani2011regression}), and DeepIV (\citealp{hartford2016counterfactual}) is more than 79\%, 57\% and 22\%, respectively. We find that a deep, multi-layered structure of DPLS significantly increases predictive performance and representation learning ability. To incorporate the uncertainty of the model coefficients, we also consider the extension of DPLS to a Bayesian framework and discuss implications.


One area for future research is the study of full uncertainty quantification. It is well known that posterior distributions for IV regression require careful assessment dating back to \cite{zellner1975bayesian} (see also \citealp{hoogerheide2007natural, lopes2010extracting}). \cite{puelz2017variable} discuss parameter uncertainty and variable selection in linear factor models.
\cite{hahn2020bayesian} provide a fully Bayesian model for posterior uncertainty for causal inference using Bayesian Additive Regression Trees (BART) to model the nonlinearity in the outcome equation. This provides a gold standard for comparison to a Bayesian DPLS method. 

The rest of the paper is outlined as follows. Section \ref{sec:model} describes our general nonlinear IV model and specific Tobit variation. In section \ref{sect_dim_red} we discuss dimensionality reduction with partial least squares. In Section \ref{sect_dpls_iv} we introduce deep partial least squares and examine the asymptotic theory. Section \ref{sect_applications} illustrates the applications of the method. Finally, Section \ref{sect_disc} concludes with directions for future research.

\section{Nonlinear IV Model}\label{sec:model}
One goal in this article is to predict the outcome $y^\star$ that possibly nonlinearly depends on a policy (treatment) $p^\star$ and predictors $x$ for each observation $i = 1 \dots, n$:
\begin{align}\label{eq_ystar}
    y^\star = p^\star\beta + x\beta_x + u.
\end{align}
We assume, neither $y^\star$ nor $p^\star$ are directly observable. Instead, they consist of errors. Define the following variables

\begin{align*}
(y^\star, y) \in \mathbb{R} =&  \; \; {\rm the \; potential \; and \; observed \; outcome \; variables},\\
(p^\star, p) \in \mathbb{R} =&  \; \; {\rm the \; potential \; and \; observed \; policy \; variables }, \\
x \in \mathbb{R}^k=&  \; \; {\rm observable \; features}, \\
z \in \mathbb{R}^m =&  \; \; {\rm instrumental \; variables}, \\
(u, w, v, \varepsilon) \in \mathbb{R}=&  \; \; {\rm latent/error \; variables \; that \; affect \;} (y^\star, p^\star, p, y), \;{\rm respectively,}  \\
m + k <& n.
\end{align*}

Instead of \eqref{eq_ystar},  we have access to the following structural equation model: 
\begin{align}\label{eq_struct}
    p = p^\star + v, \nonumber \\
    p^\star = g(z\alpha + x\alpha_z) + w,   \\
    y = \tau(y^\star) + \varepsilon. \nonumber
\end{align}
$g(\cdot)$ and $\tau(\cdot)$ are potentially non-linear continuous functions, possibly deep learners (\citealt{hartford2016counterfactual}). $\tau(\cdot)$ is a known transformation. In this study, we define the Tobit model $\tau(y^\star) = 1(y^\star > 0)\cdot y^\star$. The errors $v$ and $\varepsilon$ are assumed to be uncorrelated with any other latent error. However, policy $p^\star$ is allowed to correlate with $u$. Covariates $x$ are typically independent of $w$, $v$ and $u$. $\beta$ and $\beta_x$ represent the effects of $p^\star$ and $x$ on $y$, while $\alpha$ and $\alpha_z$ are the effects of instruments z and covariates $x$ on the policy $p^\star$, correspondingly. The errors $w$ and $u$ are correlated. For example, consider, $y$ is a customer's decision to buy an airline ticket (observed for a particular group of customers), and $p^\star$ is the price of this ticket. In that case, policy $p^\star$ is said to be endogenous when, conditional on $x$, $(p^\star, u)$ correlate. i.e, $\mathbb{E}(u|x, p^\star) \neq 0$. For example, ticket prices might increase during conferences that are unobservable to a researcher, and in that case, $\mathbb{E}(p^\star u|x) \neq 0$. Classical estimation methods, such as OLS, will lead to a spurious positive relation between prices and sales.

Define the joint error $\eta = w + v$. Moreover, from \eqref{eq_struct} we see that $p^\star = p - v$. In that case, \eqref{eq_ystar} becomes
\begin{align}\label{eq_ystar_new}
    y^\star = p \beta + x\beta_x + \xi, 
\end{align}
where $\xi = u - v\beta$. Then by combining \eqref{eq_ystar_new} and \eqref{eq_struct}, we get:
\begin{align}\label{eq_ivmodel}
    y &= \tau(p\beta + x\beta_x + \xi) + \varepsilon,  \\
    p &= g(z\alpha + x\alpha_z) + \eta, \nonumber 
\end{align}
where the latent errors $\xi$ and $\eta$ are correlated through $w$. The structural equation model in \eqref{eq_ivmodel} defines the nonlinear IV framework. In a traditional nonlinear IV setup (when $\tau(\cdot)$ is an identity function), the presence of valid instruments $z$ that satisfy Assumptions \ref{ass_rel}-\ref{ass_exog} allows us to predict the unbiased mean outcome.
\begin{assump1}[Relevance] \label{ass_rel}Instruments $z$ strongly relate to policy $p$, i.e., the density of $p$, $F(p|z, x)$, is not constant in $z$. 
\end{assump1}

\begin{assump1}[Exclusion Restriction] $z$ is conditionally orthogonal to outcome $y$: 
\begin{align*}
    z \perp y | \ (x, p, \xi).
\end{align*}
\end{assump1}

\begin{assump1}[Exogeneity]\label{ass_exog} $z$ is conditionally orthogonal to the latent error term $\xi$: 
\begin{align*}
 z \perp \xi | x.
\end{align*}
\end{assump1}


Specifically, when $y = p\beta +  x\beta_x + \epsilon$ with $\epsilon = \xi + \varepsilon$, valid instruments efficiently separate information in $p$ that is unrelated to $\xi$, and result in the consistent estimate of the effect of the policy on the outcome ($\beta$). For example, if the fuel cost represents an instrument for ticket prices, and is unrelated to conferences, it can recover the exogenous variation in ticket prices and correctly estimate the negative relation between prices and ticket sales. 

Under Assumptions \ref{ass_rel}-\ref{ass_exog}, a standard approach to estimating $\beta$ is a 2SLS method. The method entails predicting $\hat{p}$ in the first stage ("treatment network"). Then the predicted policy $\hat{p}$ replaces $p$ in the outcome equation, and we estimate a second stage regression ("outcome network") with another consistent method (\citealp{angrist1996children, hartford2016counterfactual, mogstad2021causal}). 

Define the predicted mean outcome: 
\begin{align}
    \mathbb{E}(y|x, z) = \mathbb{E}[f(p, x)|x, z] + \mathbb{E}[\xi|x].
\end{align}
Then we can evaluate the effect of a marginal change in policy (e.g., prices) from $p_0$ to $p_1$ on the outcome of interest (treatment effect):
\begin{align*}
    \mathbb{E}(y|p_1, x) - \mathbb{E}(y|p_0, x) = \mathbb{E}[f(p_1, x)|x, z] - \mathbb{E}[f(p_0, x)|x, z].
\end{align*}

The issue is that $\tau(.)$ typically is not an identity function. In such nonlinear regression models, latent errors are no longer additively separable from the true regressors ($\hat{p}$ and $u$ are still related), and hence, the true relationship breaks down with errors in variables. The orthogonality condition of the instruments and the outcome error is violated: $\mathbb{E}\big(z(y - p\beta)|x\big)\neq 0$. As a result, the 2SLS estimator fails to be consistent for nonlinear errors-in-variables models (\citealp{amemiya1985instrumental}). 
Moreover, when instruments and/or covariates are close to the number of observations and nonlinearly relate to the policy and outcomes, OLS is no longer an efficient solution of \eqref{eq_ivmodel}. 

The goal of this paper is to predict outcome $y^\star$ with the policy that depends on many instruments $z$ and regressors $x$. To do so, we extend deep partial least squares with a recentered and rescaled outcome (\citealp{iwata2001recentered}) and additionally illustrate it with a control function approach (\citealp{terza2008two}).

\subsection{Predicton in Low-Dimensional Data}

In this section, our attention is directed towards predicting the outcome using a Tobit model, wherein $\tau(y^\star) = 1(y^\star > 0)y^\star$. We demonstrate that by recentering and rescaling the outcome, two-stage least squares (2SLS) can yield a consistent estimator of $\beta$. Alternatively, the control function approach can provide a valid prediction of the outcome.

\subsubsection{2SLS with the Tobit Model}
To establish consistency, we impose an essential assumption regarding the underlying process responsible for generating the observed data. 

\begin{assump1}[Elliptical Distribution]\label{ass_elliptic}
$ ( y^\star, z ) $ have a joint elliptical (or Gaussian in the simplest setting) distribution.
\end{assump1}

Assumption \ref{ass_elliptic} requires that the instrumental variables and the outcome are jointly elliptically distributed. The assumption can be strong in many settings. Though \cite{brillinger2012generalized} and \cite{iwata2001recentered} show that the results are robust to significant deviations from the assumption. This assumption can also be relaxed and generalized to other distributions (see e.g. \citealp{adcock2007extensions}).

Define the recentered and rescaled output variable that mimics the unobserved outcome $y^\star$,
$$
\tilde{y} = \psi_1^{-1} ( y - \psi_2 ), \; \; {\rm where} \; \; \psi_1 = \mathrm{cov}(y,y^\star ) /\mathrm{var}(y^\star) , \psi_2 = \mathbb{E}(y) -\psi_1 \mathbb{E}(y^\star ) .
$$
Specifically, $\psi_1$ is the ratio of the covariance between $y$ and $y^\star$ to the variance of $y^\star$, while $\psi_2$ is the intercept term when we consider a linear projection of $y$ on $y^\star$.  Next, consider a linear projection of z on $y^\star$

\begin{align*}
    \mathbb{E}(z|y^\star) = \mathbb{E}(z) + \frac{\mathrm{cov}(z, y^\star)}{\mathrm{Var}(y^\star)}\big[y^\star - \mathbb{E}(y^\star)].
\end{align*}

Hence,

\begin{align}\label{eq_steiner}
\mathbb{E}( z \mid y^\star ) - \mathbb{E}(z) = \gamma ( y^\star - \mathbb{E}( y^\star ) ),
\end{align}

where $\gamma = \mathrm{cov}(z, y^\star)/\mathrm{var}(y^\star)$ \footnote{Elliptical contours allow for discrete outcomes.}.  By Stein's lemma \citep{landsman2008stein},  we can calculate the covariance of instruments and rescaled output as 
\begin{align}\label{eq_replace}
\mathrm{cov}( z , \tilde{y} )   = \psi_1^{-1} cov( z,y) = \psi_1^{-1}  \mathbb{E}_{ y^\star } \left ( \mathbb{E}\big[(  z - \mathbb{E}(z) ) y \mid y^\star \big]  \right ). 
\end{align}

Conditional expectations are linear under elliptical contours. Additionally, $ \mathbb{E}( y \mid y^\star ) = \tau( y^\star ) $, $\gamma = \mathrm{cov}(z, y^\star) / var(y^\star)$, and replacing  the last equality from \eqref{eq_replace} with \eqref{eq_steiner} gives

\begin{align*}
\mathrm{cov}( z , \tilde{y} ) & =  \psi_1^{-1}  \mathbb{E}_{ y^\star } \left ( \mathbb{E}( \gamma ( y^\star - \mathbb{E}( y^\star ) )  y \mid y^\star   \right )  \\
&= \gamma \psi_1^{-1} \mathbb{E}(  ( y^\star - \mathbb{E}( y^\star ) )  y ) = \gamma \psi_1^{-1} \mathrm{cov}( y^\star , y  ) \\
& = \psi_1^{-1} \frac{cov(z, y^\star)}{var(y^\star)} \mathrm{cov}(y , y^\star )  = \psi_1^{-1}cov(z, y^\star)\psi_1 = cov(z, y^\star).
\end{align*} 

For additional details, see \cite{iwata2001recentered}. 2SLS first regresses $ p$ on $ z$ and $x$ to get the predicted policy $\hat{p}$. Define $\bar{Z} = [z, x]$ and  $\hat{p} = P_z \bar{Z}  $ where $ P_z = \bar{Z}( \bar{Z}^T \bar{Z} )^{-1} \bar{Z}^T $ is the projection matrix onto the instrumental variable space (including covariates). In addition, define $\bar{P} = [\hat{p}, x]$ and let $ e = \tilde{y} - p \beta $ be the residual from the re-scaled regression, then
$$
\mathrm{cov}(z,e | x) = \mathrm{cov}( z , \tilde{y} - p\beta |x ) = \mathrm{cov}( z , y^\star - p\beta | x ) = \mathrm{cov}(z,u | x) =  0.
$$
The re-scaled instrumental variable estimator is given by
\begin{align}\label{eq_gmm}
 \hat{\beta}^{GMM} = ( \bar{Z}^T P_z \bar{Z} )^{-1} P_z \bar{Z}^T \tilde{y} = ( \bar{P}^T \bar{P} )^{-1} \bar{P}^T \tilde{y}    
\end{align}

From the above, we then have $ \mathbb{E}( \hat{\beta}^{GMM} ) = \beta_p $ with $\beta_p = [\beta, \beta_x]$ and $ \mathrm{E}( \tilde{y} ) = \bar{P} \beta_p $. Note that $ \hat{\beta}^{GMM}$ is a $(m+k) \times 1$ vector, where the first element is the policy effect on outcome.

\subsubsection{Estimating the Proportionality Constants}

In the Tobit case, the constants $\phi_1, \phi_2$ can be calculated theoretically from the model as
$$
\phi_1 = \Phi \left ( \delta \right )  \; \; {\rm and} \; \; \phi_2 =  \sigma^\star_y \phi ( \delta ), 
$$
where $\Phi \left ( \delta \right )$ and $\phi ( \delta )$ are the cumulative and probability density functions of a normally distributed random variable, respectively, and 
$$
 \delta = \mu_X^T \beta / \sigma^\star_y \; \; {\rm and} \; \;  ( \sigma^\star_y  )^2 = \sigma^2_{y^\star } + \beta^T \Sigma_{xx} \beta .
 $$
\cite{iwata2001recentered} shows that, 
with the truncated outcome variable, we can obtain their estimators:
\begin{align*}
\hat{\psi}_1 = \hat{\Phi} = \dfrac{1}{n}\sum_{i=1}^n\mathds{1}(y >0), \\
\hat{\psi}_2 =  \hat{\sigma}^\star_y \hat{\phi } \; \; {\rm where} \; \hat{\phi } =  \phi \big( \Phi^{-1} ( \hat{\psi}_1 ) \big).
\end{align*}

The estimator of the variance of the outcome variable is given as
\[
 \hat \sigma_y^{\star 2} = \dfrac{1}{n \; c( \hat{k} ) }\sum_{i=1}^n(y_i - \bar y)^2,
\]
where 
$$ 
c( \hat{k} ) = \hat{\Phi} - ( \hat{\phi} - \Phi^{-1} ( \hat{\Phi} )  ( 1 - \hat{\Phi} ) ) ( \hat{\phi} +  \Phi^{-1} ( \hat{\Phi} )\hat{\Phi} ),
$$
and $\hat{\Phi} = \hat{\psi}_1 $ and $\phi$ are the cumulative and probability density functions of a normally distributed random variable, respectively. $\mathds{1}(y >0)$ is a binary variable and equals one if the outcome is positive.
\cite{iwata2001recentered} shows that the scaling constant $\widehat{\psi}_1$ and the estimator $\hat{\beta}^{GMM}$ defined in \eqref{eq_gmm} are asymptotically jointly normally distributed, with
\begin{align}\label{eq_asympt}
    \sqrt{n}(\hat{\psi}_1 - \psi_1) \rightarrow \mathcal{N}\big(0,\hat{\psi}_1 (1-\hat{\psi}_1)\big), \\
    \sqrt{n}(\hat{\beta}^{GMM} - \beta) \rightarrow \mathcal{N}\big(0, \Sigma_{\star} - \Phi(1-\Phi)\beta\beta^T\big).
\end{align}
Based on \cite{hansen1982generalized},  $\Sigma_{\star}$ is the variance of the standard GMM estimator and can be computed as follows \footnote{See also \cite{baum2003instrumental} for the discussion of a GMM estimator for the IV regression. }:
\begin{align*}
    \widehat{\Sigma}_{\star}& = n(\bar{P}^T\bar{Z}\widehat{G}\bar{Z}^T\bar{P})^{-1}\bar{P}^T\bar{Z}\widehat{G}(n\widehat{A}) \widehat{G}\bar{Z}^T\bar{P}(\bar{P}^T\bar{Z}\widehat{G}\bar{Z}^T\bar{P})^{-1}, \\
    \widehat{A} &= \frac{1}{n}\sum_i \widehat{e}^2_i \bar{Z}_i\bar{Z}_i^T, \ \  \text{with} \ \  \widehat{e}^2 = \widehat{\tilde{y}} - \bar{P}\hat{\beta}^{GMM}, \\
    \widehat{G} & = \widehat{H}\big( \widehat{H}^{-1} - \bar{Z}^T\bar{Z} /\bar{Z}^T\widehat{H}\bar{Z}\big)\widehat{H}, \\
    \widehat{H} & = \widehat{A}^{-1}.
\end{align*}

\subsection{Control Function Approach}
An alternative approach to estimate the parameters is a control function approach \citep{terza2008two}. Instead of substituting the predicted policy in the second stage regression, we control for the predicted residuals of the first stage:

\begin{align}
    \widehat{\eta} &= p - \widehat{p}, \\
    y &= \tau(p\beta + \widehat{\eta}\beta_\eta) + \varepsilon^{2SRI}, 
\end{align}
where the residuals of the treatment network $\widehat{\eta}$ can be predicted by any consistent method. Note that $\varepsilon^{2SRI}$ is not identical to $\varepsilon$ due to the substitution of $\xi$ with $\widehat{\eta}$.  The control function approach can be viewed as a special case of 2SLS. Specifically, the inclusion of $\widehat{\eta}$ in the outcome equation of \eqref{eq_ivmodel} allows us to control for the correlation of $\eta$ and $\xi$, and predict the outcome. 

\section{Dimensionality Reduction}\label{sect_dim_red}
The estimation of $\hat{\beta}^{GMM}$ in \eqref{eq_gmm} rests on the assumption that the number of instruments (as well as independent variables) is strictly lower than the number of observations. Throughout the proof, we maintain this assumption. 

\begin{assump1}
The number of covariates is strictly smaller than the number of observations, $m + k < n$, where $m$ is the dimension of the instruments and $k$ is the number of covariates. 
\end{assump1}

Nevertheless, when the dimension of the independent variables is close to $N$, $\hat{\beta}^{GMM}$ captures the unwanted variation reflected in the predicted policy $\hat{p}$ (\citealp{gagnon2013removing}). Since $\bar{Z}$ are independently and identically distributed elliptical random variables with covariance $\Sigma_{zz}$, \cite{brillinger2012generalized} shows that 
\begin{align}\label{eq_brilling}
    \text{cov}(\bar{Z}, p) = \text{cov}(\bar{Z}, f(U)) = \text{cov}(\bar{Z}, U) \text{cov}(f(U), U)/\text{var}(U) = k\Sigma_{zz}\alpha_z,
\end{align}
where $\bar{Z} = [z, x]$, $U = \bar{Z}\alpha_z$ with $\alpha_z = [\alpha, \alpha_x]$,  and the constant $k = \text{cov}(f(U), U)/\text{var}(U)$. Based on this result, \cite{brillinger2012generalized} shows that the OLS coefficient is a consistent estimate of $\alpha$ up to a proportionality constant. In this article, we show that the deep partial least squares method has the same property.

\subsection{Partial Least Squares}
Partial least squares is a dimensionality reduction method that generalizes and combines features from the principal component analysis and multiple regression (\citealp{abdi2010principal}). 

Consider the augmented instrumental variable $\bar{Z} = [z, x]$ with the dimension $n \times (m + k)$, where $m$ and $k$ are the dimensions of instruments and covariates, respectively. The policy $p = P$ is a $n \times 1$ vector as before. PLS can be summarized by the following relationship:
\begin{align}
    \bar{Z} = TV + F, \\
    P = UQ + E, \nonumber
\end{align}
where $T$ and $U$ are $n \times L$ projections (scores) of $\bar{Z}$ and $P$, respectively. $V$ and $Q$ are orthogonal projection matrices (loadings). Maximizing the covariance between the augmented instruments and the policy leads to the first PLS projection pair $(v_1, q_1)$:
\begin{align*}
  & \max_{v, q} \ (\bar{Z}v_1)^T(Pq_1) \\ 
  & \text{subject to} \ ||v_1|| = ||q_1|| = 1. 
\end{align*}
The corresponding scores are $t_1 = \bar{Z}v_1$ and $u_1 = Pq_1$. It is clear that the directions (loadings) for the policy $P$ and the augmented instruments $\bar{Z}$ are the right and left singular vectors of $\bar{Z}^TP$, respectively. PLS in the next step performs an ordinary regression of $U$ on $T$, namely $U = T\beta$. Then the next projection pair $(v_2, q_2)$ is found by calculating the singular vectors of the residual matrix $(\bar{Z} - t_1v_1^T)^T(P - T\beta q^T)$. Lastly, the final regression of interest is $U = T\beta$ (the tutorial \cite{geladi1986partial} contains further details).

The key property of PLS is that it is consistent for estimating parameters of interest even in the presence of nonlinearity via a sequence of covariance calculations. \cite{fisher1922mathematical} first observed this in the Probit regression. \cite{naik2000partial} show the consistency of PLS based on the result of \cite{brillinger2012generalized}.  One desirable property of PLS is that it has a closed-form solution. The PLS estimator of $\alpha$ is given as follows (\citealp{helland1990partial, stone1990continuum}): 

\begin{align}
    \widehat{\alpha}^{PLS} = \hat{R}(\hat{R}^TS_{zz}\hat{R})^{-1}\hat{R}^Ts_{zp},
\end{align}
where $\hat{R} = (s_{zp}, S_{zz}s_{zp}, \dots, S_{zz}^{q-1}s_{zp})$ is the $(m+k) \times q$ matrix of the Krylov sequence with a $(m+k) \times (m+k)$ matrix $S_{zz}$  and a $(m+k) \times 1$ vector $s_{zp}$ defined as follows:
\begin{align*}
    S_{zz} = \frac{\bar{Z}^T(I - 11^T/n)\bar{Z}}{n-1}, \\
    s_{zp} =  \frac{(\bar{Z} - \mathbb{E}(\bar{Z}))^T(P - \mathbb{E}(P))}{n-1},
\end{align*}
where $I$ is an identity matrix and $1$ is a matrix of ones. Intuitively, PLS searches for factors that capture the highest variability in $\bar{Z}$, and at the same time maximizes the covariance between $\bar{Z}$ and $P$. If the number of factors equals the dimension of instruments, $q = m+k$, the method is equivalent to OLS (\citealp{helland1990partial}).

\section{Deep Partial Least Squares for IV Regression}\label{sect_dpls_iv}

A useful generalization of PLS is to consider a deep-layered feed-forward neural network structure:
\begin{align} \label{eq_dnn}
    \hat{p}^{(1)} &= f(\bar{Z}\hat{\alpha}^{PLS}) \nonumber, \\
    \hat{p}^{(2)} &= f(\hat{p}^{(1)}\hat{\alpha}^{(2)}), \nonumber\\
   & \vdots \\
    \hat{p}^{(L)} &= f(\hat{p}^{(L-1)}\hat{\alpha}^{(L)}), \nonumber
\end{align}
where $f(\cdot) = \max(\cdot, 0)$ is a rectified linear unit (ReLU) activation function. Note that $\bar{Z} = [z, x]$ as before and $\hat{\alpha}^{PLS}$ is a $q \times 1$ vector where $q$ is the number of PLS factors in the treatment network. To reduce the dimensionality of the instruments, we predict the treatment in the first layer by PLS. The parameters $\alpha^{(\ell)}$ in subsequent layers $\ell = 2, \dots, L$ can be identified by OLS or PLS. The following proposition shows that the parameters are consistent in each layer up to a proportionality constant.  

\begin{proposition} \label{prop_1}
Let $S_{zz}$ and $s_{zp}$ converge in probability to $\Sigma_{zz}$ (the population variance of $z$) and $\sigma_{zp}$ (the population covariance of $z$ and $p$) when $n \rightarrow \infty$. Moreover, let there exist a pair of eigenvectors and eigenvalues $(v_j, \lambda_j)$ for which $\sigma_{zp} = \sum_{j=1}^M\gamma_jv_j$ (with $\gamma_j$ non-zero for each $j = 1, \dots, M$).  Assume also $\mathbb{E}(|g(U)|) < \infty$ and $\mathbb{E}(U|g(U)|) < \infty$ and q = M. Then $\hat{\alpha} = \{ \hat{\alpha}^{PLS}, \hat{\alpha}^{(2)}, \dots, \hat{\alpha}^{(L)}\}$ are consistent up to a proportionality constant. 
\end{proposition}

\begin{proof}
We follow the approach by \cite{naik2000partial}. Let $\alpha^\star = \Sigma_{zz}^{-1}\sigma_{zp}$. Define $R = (\sigma_{zp}, \Sigma_{zz}\sigma_{zp}, \dots, \Sigma_{zz}^{q-1}\sigma_{zp})$. Then, based on the assumption that $S_{zz} \rightarrow \Sigma_{zz}$ and $s_{zp} \rightarrow \sigma_{zp}$  when n $\rightarrow \infty$, we have:
\begin{align*}
    \hat{\alpha}^{PLS} \rightarrow R(R^T\Sigma_{zz}R)^{-1}\Sigma_{zz}\alpha^\star \ \text{ in probability when} \ n \rightarrow \infty. 
\end{align*}
The assumptions $q = M$ and $\sigma_{zp} = \sum_{j=1}^M\gamma_jv_j$ imply that $\alpha^\star$ is contained in the space spanned by $R$. Consequently, $\Sigma_{zz}^{1/2}\alpha^\star$ is contained in the space spanned by $R^\star = \Sigma_{zz}^{1/2}R$. Therefore,
\begin{align*}
    R^\star(R^{\star^T}R^\star)^{-1} R^{\star^T}\Sigma_{zz}^{1/2}\alpha^\star = \Sigma_{zz}^{1/2}\alpha^\star,
\end{align*}
and 
\begin{align*}
    R(R^{\star^T}R^\star)^{-1} R^T\Sigma_{zz}\alpha^\star = \alpha^\star.
\end{align*}

Hence, $\hat{\alpha}^{PLS} \rightarrow \alpha^\star$. Equation \eqref{eq_brilling} implies that $\sigma_{zp} = k\Sigma_{zz}\alpha$, and therefore, $\hat{\alpha}^{PLS} \rightarrow \Sigma_{zz}^{-1}\sigma_{zp} = k\alpha$. $k$ is defined in \eqref{eq_brilling}. This proves that the PLS in the first layer is consistent. 

Now, consider the second layer $\ell = 2$. If $\hat{\alpha}^{(2)}$ is estimated by either OLS or PLS, \eqref{eq_brilling} directly imply that $\hat{\alpha}^{(2)} = k\alpha$. Next, to run OLS (or PLS) of $p$ on $\hat{p}^{(1)}$, we construct 
\begin{align*}
 \hat{p}^{(1)} = \max(z\hat{\alpha}^{PLS}, 0) = \max(zk\alpha^{(1)}, 0) = k\max(z\alpha^{(1)}, 0) = kz^{(1)}. 
\end{align*}
The model then becomes
\begin{align}
    \hat{p}^{(2)} = max(kz^{(1)}k_2\alpha^{(2)}, 0) = k\cdot k_2 max( z^{(1)}\alpha^{(2)}, 0) = k^{(2)}z^{(2)}.
\end{align}

By induction, we can consistently estimate $\alpha$ in each subsequent layer, up to $L$. See \cite{polson2021deep} for the detailed discussion of such a deep learning structure. 

Lastly, we note that the consistent estimator of the effect of the treatment on the outcome is defined as
\begin{align}
    \hat{\beta}^{GMM} =  \big((\bar{P}^{(L)})^T\bar{P}^{(L)}\big)^{-1}(\bar{P}^{(L)})^T\tilde{y},
\end{align}
where $\bar{P}^{(L)}$ is the augmented treatment $[\hat{p}, x]$ predicted by layer $L$. Moreover, the predicted outcome is $\hat{y} = \bar{P} \hat{\beta}^{GMM}$.

Define $\hat{e}^L$ as the predicted residuals from the treatment network in layer $L$ and consider, $\tilde{p} = [p, \hat{e}^L, x]$. Then, similarly for 2SRI: 
\begin{align}
    \hat{\beta}^{GMM} = (\tilde{p}^T\tilde{p})^{-1}\tilde{p}^Ty.
\end{align}
and $\hat{y} = \tilde{p}\hat{\beta}^{GMM}$.

\end{proof}

\subsection{Prediction and Bayesian Shrinkage}
One possible extension of deep partial least squares is a quantification of uncertainty in the density of the outcome. Our probabilistic model takes the form:
\begin{align*}
\tilde{y} \mid f,\bar{P} &  \sim p( \tilde{y} \mid f, \bar{P} ), \\
f & = g ( \bar{P} \beta_p ) + \varepsilon ,
\end{align*}
where $\tilde{y}$ is an $n \times 1$ rescaled and recentered outcome variable as before, $\bar{P} = [\hat{p}, x]$ is an $n \times (1+k)$ augmented (predicted) policy.  $\beta_p = [\beta, \beta_x]$ is a $(1+k) \times 1$ vector of coefficients in the outcome network. Here $g$ is a deep partial least squares method. To estimate parameters in the first layer, the method uses the SIMPLS algorithm (\citealp{de1993simpls}). Subsequent layers use the stochastic gradient descent (SGD) method (\citealp{bottou2012stochastic}) for optimizing and training the parameters.  

The key result, due to \cite{brillinger2012generalized} and \cite{naik2000partial}, is that $\beta_p$ can be estimated consistently, up to a constant of proportionality using
PLS, irrespective of the nonlinearity of $g$. Given a specification of $g$, the constant of proportionality can also be estimated consistently with $ \sqrt{n} $-asymptotics. It is worth noting that typically, standard SGD methods  will not yield asymptotically normally distributed parameters. However, they can substantially increase the precision of the coefficients of interest. 

Suppose that we wish to predict the outcome at a new level $ \bar{P}^\star $. Then, we can use the predictive distribution to make a forecast as well as provide uncertainty bounds:
$$
y_\star \sim p \left ( y \; \mid \;   g (  \bar{P}^\star\hat{\beta}_{p}^{DPLS} )  \right ) .
$$

The advantage of modeling a probabilistic model is the flexibility and  the possibility to incorporate uncertainty in the parameters of interest. We approximate the posterior distribution of $\hat{\beta}^{GMM}$ with its asymptotic distribution based on the Bernstein-Von Mises theorem (see e.g.,  \citealp{van2000asymptotic, bhadra2019lasso}). Define $\text{Data} = (\tilde{y}, p, z, x) $, then the densities $P(\hat{\beta}^{GMM}| \text{Data})$ and $P(\hat{\psi}_1)$ come from a normal distribution with the mean and variance depicted in \eqref{eq_asympt}. To shrink the effect of redundant instruments in the treatment network, we consider a Ridge regression estimator (\citealp{marquardt1975ridge}):
\begin{align}
    \hat{\alpha}^{Ridge} = (\bar{Z}^T\bar{Z} + \lambda I_m)^{-1}\bar{Z}^T p. 
\end{align}

However, as pointed out by a referee, this still underestimates the uncertainty due to the estimation of $ ( \hat{\beta}_{p}^{DPLS}, \hat{\alpha}^{Ridge} ) $. A fully Bayesian model with a non-uniform prior density of the coefficients can increase the precision of uncertainty bounds.

\section{Applications} \label{sect_applications}
In this section, we evaluate the predictive performance of the DPLS-IV method relative to benchmark methods. We generate synthetic data  to mimic the high-dimensional and sparse nature of instruments as described by \cite{kong2018deep}. To illustrate how well the method captures the complex nonlinear and sparse relation of covariates, our design mimics their setup. In addition to synthetic designs, this section illustrates findings on data provided by \cite{angrist1996children}. Throughout the experiments, we split data into two partitions. To find optimal parameters in each method, we use one sub-sample (further partitioned into train and validation data) and illustrate prediction performance measures on the other. We implement the experiments in R and the code to replicate our findings is available upon request.

\subsection{Sparse Uncorrelated Instruments}
The data-generating process is summarized by the following structural equation model:
\begin{align} \label{eq_simu}
   y &= f(p\beta + x\beta_x + \xi) + \varepsilon, \\
      p &= g(z\alpha + \text{sigmoid}(z^2)\gamma + x\alpha_x) + w, \nonumber \\  
    & w; u  \sim  \mathcal{N}(0, \Sigma) \nonumber, \nonumber \\
    & \varepsilon  \sim  \mathcal{N}(0, \sigma_\varepsilon^2) \nonumber, \nonumber \\
    z & \sim \mathcal{N}(0, \Sigma_z), \ x \sim N(0, \Sigma_x)     \nonumber, \\
    &\alpha; \alpha_x, \gamma, \beta, \beta_x \sim \mathcal{N}(0, 1).\nonumber
\end{align}
\eqref{eq_simu} holds for each observation $i = 1, \dots n$, where $n = 1000$. $w$ and $\xi$ are correlated and jointly normally distributed with a mean vector $0$ and a variance-covariance matrix $\Sigma = \begin{pmatrix} 3.000 & -0.087 \\ -0.087 & 0.010 \end{pmatrix}$. In this simulation setup, $g$ and $f$ are modeled by ReLU (or leakyReLu) (\citealp{dubey2019comparative}). In addition to a $n \times 50$ matrix of instruments $z$, the treatment $p$ contains nonlinear transformations of $z$ based on the sigmoid function. To introduce sparsity in the design, out of fifty instruments, ten are redundant. In particular, only forty instruments are relevant for predicting the treatment $p$. $x$ is an $n \times 25$ matrix of covariates, and 20 of them have no influence on outcome $y$. $\Sigma_z$ and $\Sigma_x$ represent covariance matrices of $z$ and $x$, respectively. In this setting, the covariance between the instruments is small (0.001) and the features $x$ are uncorrelated.

Figure \ref{fig_firststage} visualizes predictions of the policy $p$ in the treatment network  (Figure \ref{fig_secondstage} in Appendix \ref{app_leakyrelu} shows predicted outcome $y$).   In each case, $g$ and $f$ represent a rectified linear unit function.  A visual inspection of Figures  \ref{fig_firststage} and \ref{fig_secondstage} verifies that DPLS-IV results in more accurate predictions relative to the other methods. In this setting, DeepIV (\citealp{hartford2016counterfactual}) is the second-best alternative. According to Figures \ref{fig_firststage} and \ref{fig_secondstage}, predictions of the treatment, $\widehat{p}$, appear to have a higher variability compared to the outcome predictions. This is not surprising, as $p$ contains instruments in addition to $x$. Note that, even though prediction performance measures of DeepIV are close to those of DPLS-IV, Table \ref{tab_networks_lrelu} in Appendix \ref{app_leakyrelu} shows that, compared to DeepIV, DPLS-IV is more robust to changes in the activation function.

To evaluate the effectiveness of different methods, we examine their out-of-sample $R^2$ and root mean squared error (RMSE) as we increase the parameter values of $\Sigma$. The results, as depicted in Figure \ref{fig_cov_vs_prediction}, demonstrate that DPLS-IV exhibits robustness to increasing errors and is capable of accounting for the endogeneity of $p$ reflected in the covariance of error terms $w$ and $\xi$. Furthermore, our analysis, as shown in Table \ref{tab_networks}, indicates that DPLS-IV outperforms OLS, PLS, LASSO, and DeepIV methods by a significant margin. These findings suggest that DPLS-IV holds considerable promise as a powerful and reliable tool for predicting outcomes in the presence of endogeneity and measurement error.

\begin{figure}[H]
	\centering
	\subfloat[OLS]{\includegraphics[width=0.4\linewidth]{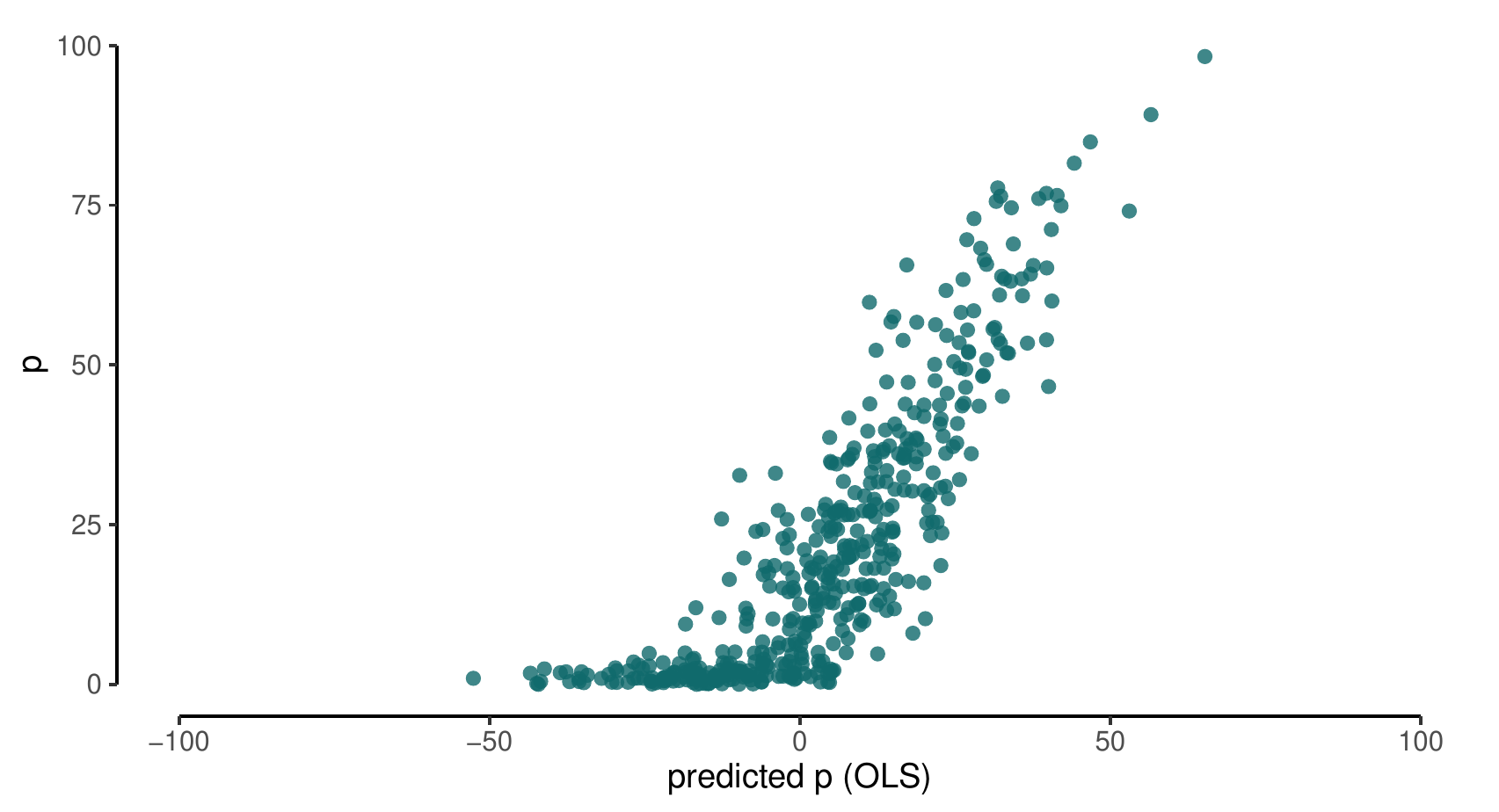}}\quad
	\subfloat[PLS]{\includegraphics[width=0.4\linewidth]{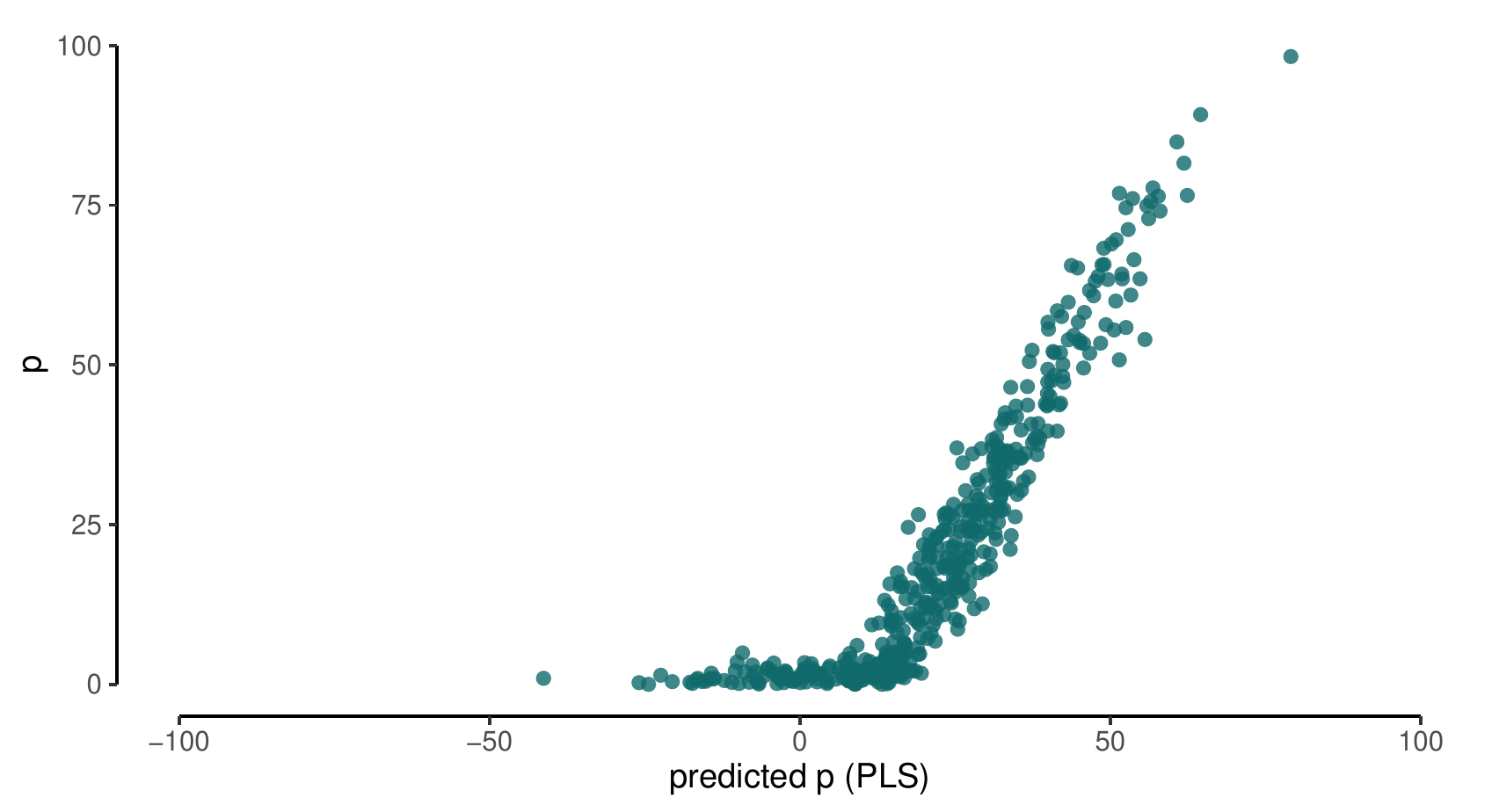}} \\
	\subfloat[DeepIV]{\includegraphics[width=0.4\linewidth]{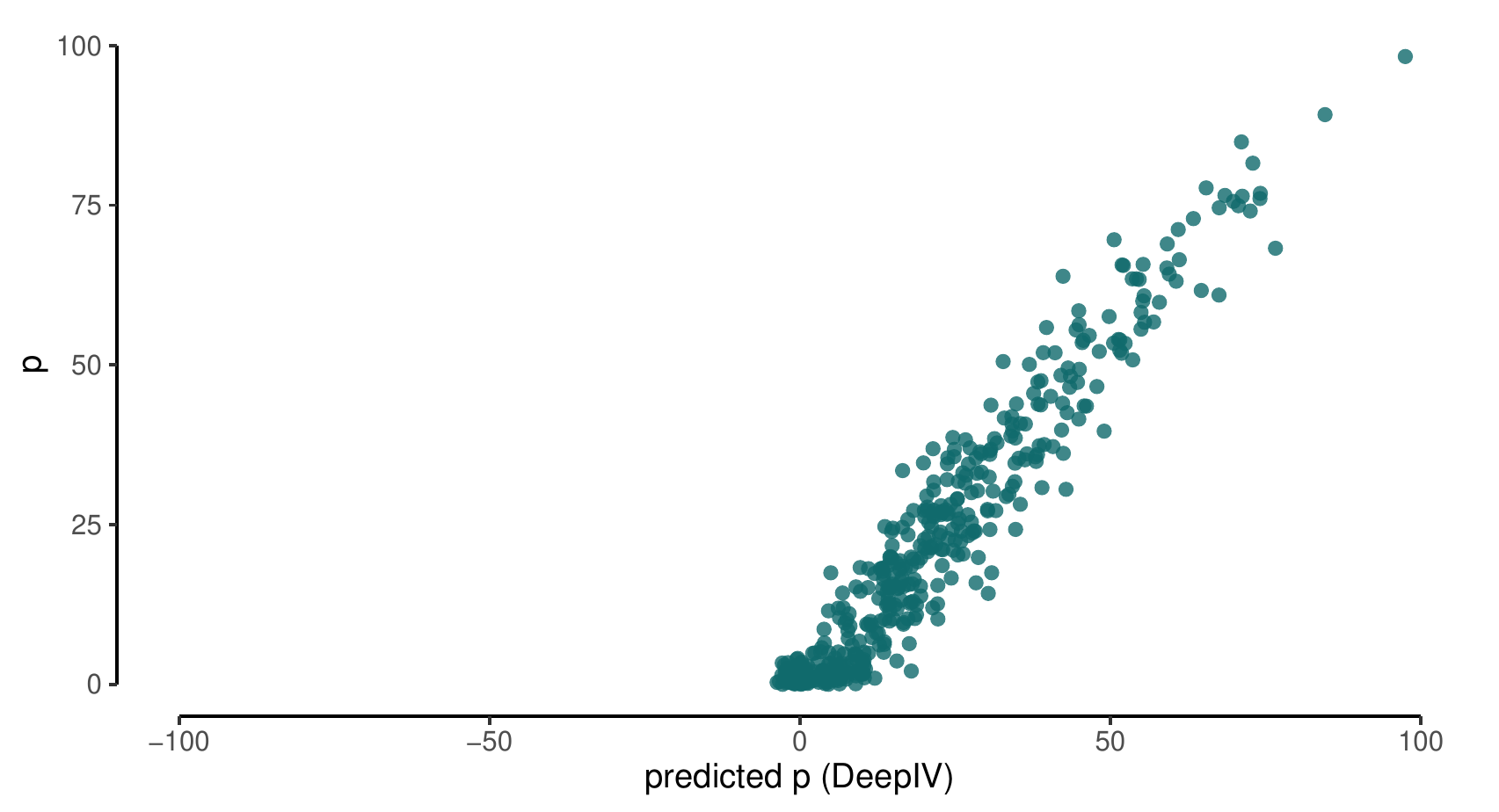}} 
	\subfloat[LASSO]{\includegraphics[width=0.4\linewidth]{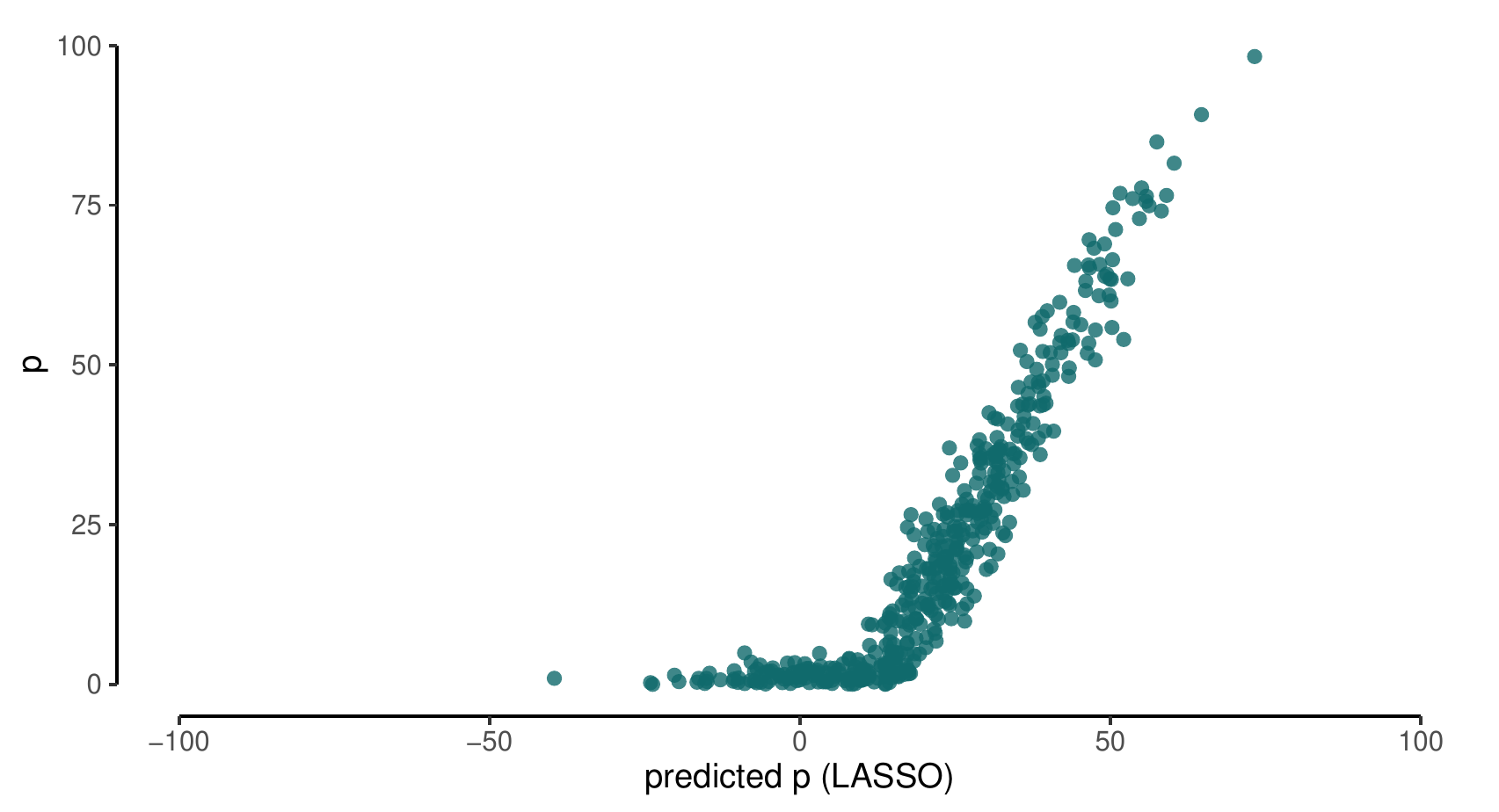}}\quad\\
	\subfloat[DPLS-IV]{\includegraphics[width=0.4\linewidth]{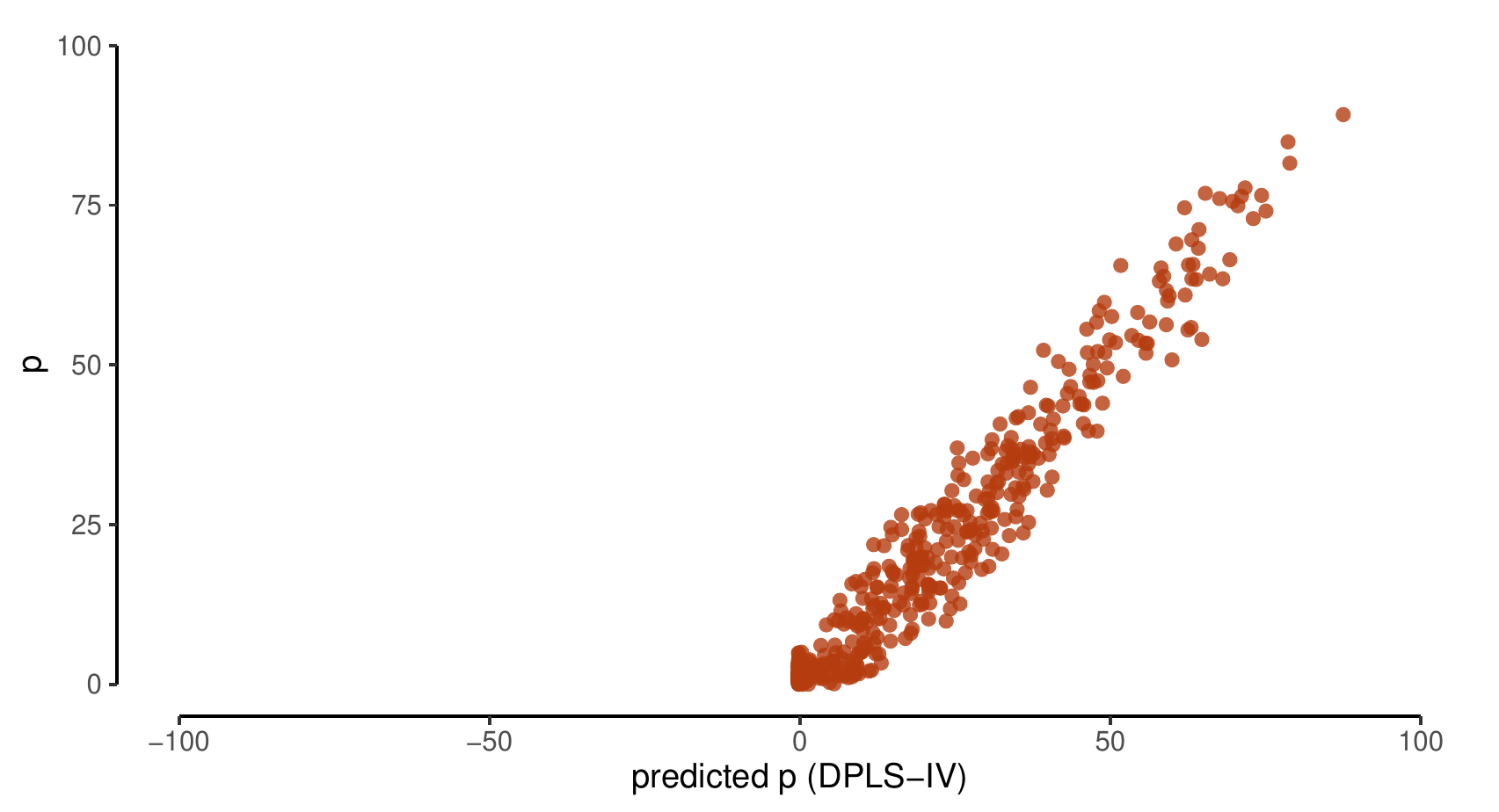}} \\
	\caption{First stage prediction performance. The X-axis depicts predicted treatment $\widehat{p}$, and the Y-axis represents true values of $p$. DPLS-IV denotes the method introduced in this study. We use test data for evaluating the methods. }\label{fig_firststage}
\end{figure}

To illustrate the predictive power of DPLS-IV, Figure \ref{fig_coeff_sim1}  compares coefficients estimated by the OLS and PLS methods in the treatment network. Parameters estimated by PLS are closer to their true values relative to OLS. Additionally, Figure \ref{fig_cdf_sim1} shows the absolute bias of these parameters. According to Figure \ref{fig_cdf_sim1}, the cumulative distribution function of the absolute bias of the parameters recovered by PLS stochastically dominates the ones based on OLS and LASSO. Specifically, smaller values of the absolute bias of the coefficients are more likely under PLS, relative to OLS and LASSO. The sum of the absolute bias is the smallest under PLS (14.406), followed by LASSO (15.233). OLS leads to the highest value of the sum of the absolute bias (21.610).   

\begin{table}[H] \centering 
    \resizebox{13cm}{!}{

\begin{tabular}{@{\extracolsep{5pt}}lccccccc} 
\\[-1.8ex]\hline 
\hline \\[-1.8ex] 
Treatment network \\
\hline \\[-1.8ex] 
Measures & \multicolumn{1}{c}{PLS} & \multicolumn{1}{c}{OLS}  &  \multicolumn{1}{c}{DeepIV} &\multicolumn{1}{c}{LASSO}  & \multicolumn{1}{c}{DPLS-IV} \\ 
\hline \\[-1.8ex] 
$R^2$ & 0.751 & 0.688  & 0.932 & 0.753 &\textbf{0.956} & \\ 
RMSE & 10.603 & 21.540  & 5.789 & 10.497 & \bf{4.508} & \\ 
\hline \\[-1.8ex] 
Outcome network \\
\hline \\[-1.8ex] 
$R^2$ & 0.932 & 0.933  & 0.889 & 0.933 & \textbf{0.938}  \\ 
RMSE & 1.668 & 1.718  & 4.573 & 1.670 & \bf{1.622}  \\ 
\hline \\[-1.8ex] 
\end{tabular}}
\caption{Prediction performance of DPLS-IV relative to other methods. We present out-of-sample $R^2$ and RMSE. To present the maximum prediction performance of OLS, PLS and LASSO in the outcome network, we use residuals predicted by DPLS-IV in the first stage.}\label{tab_networks} 
\end{table}

\begin{figure}[H]
\centering
\includegraphics[width=0.6\linewidth]{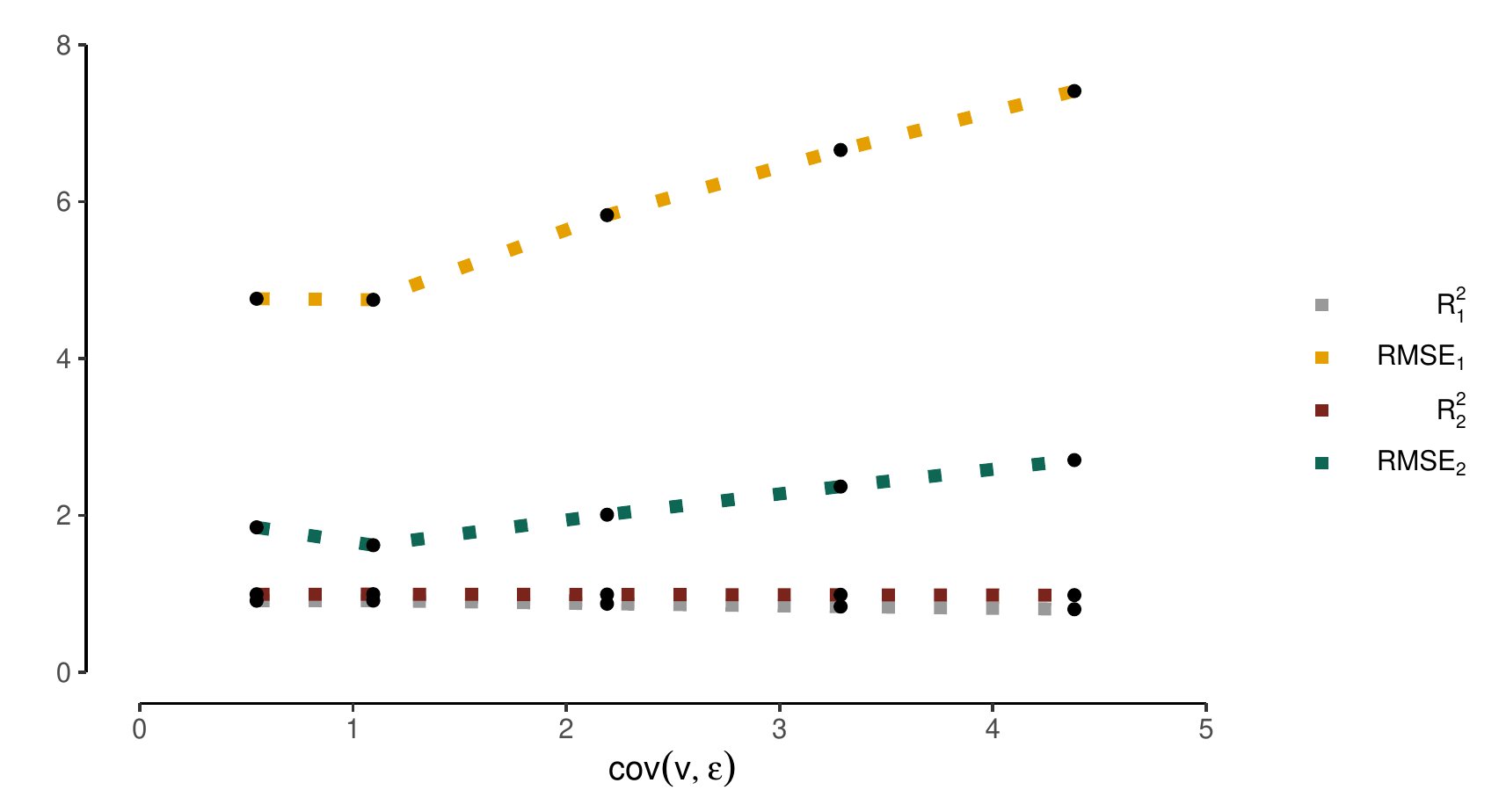}
\caption{$\text{R}^2$ and RMSE for increasing values of $\text{cov}(v, \varepsilon)$. The black circles represent the estimated values which are combined with the dotted lines. Each color corresponds to the corresponding $\text{R}^2$ and RMSE in the treatment and outcome networks. Specifically,   $\text{R}^2_1$ and $\text{RMSE}_1$ denote prediction performance measures in a treatment network. $R^2_2$ and $\text{RMSE}_2$ are the prediction performance measures in the outcome network.} \label{fig_cov_vs_prediction}
\end{figure}

\begin{figure}[H]
	\centering
	\subfloat[OLS]{\includegraphics[width=0.4\linewidth]{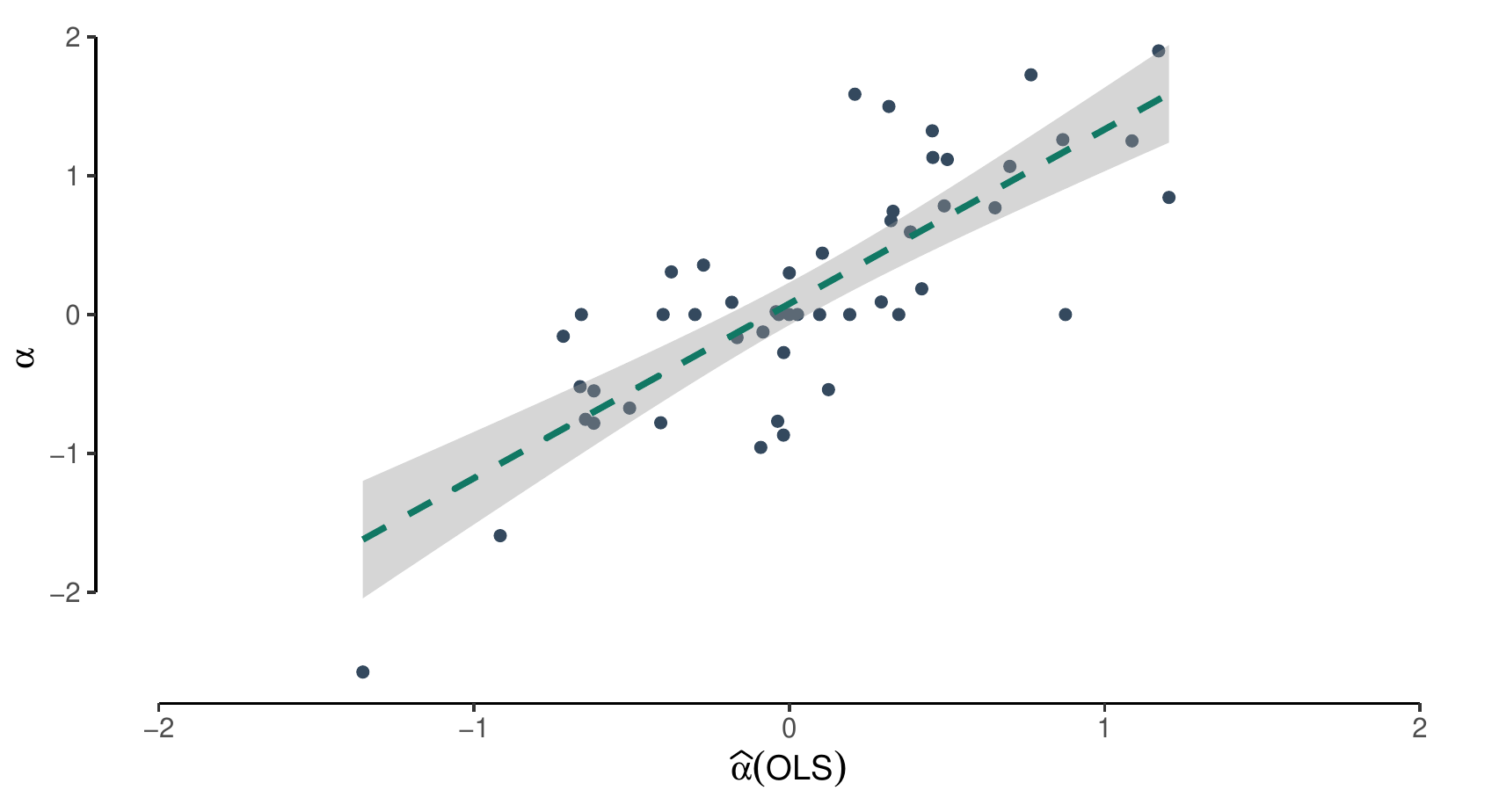}}\quad
	\subfloat[PLS]{\includegraphics[width=0.4\linewidth]{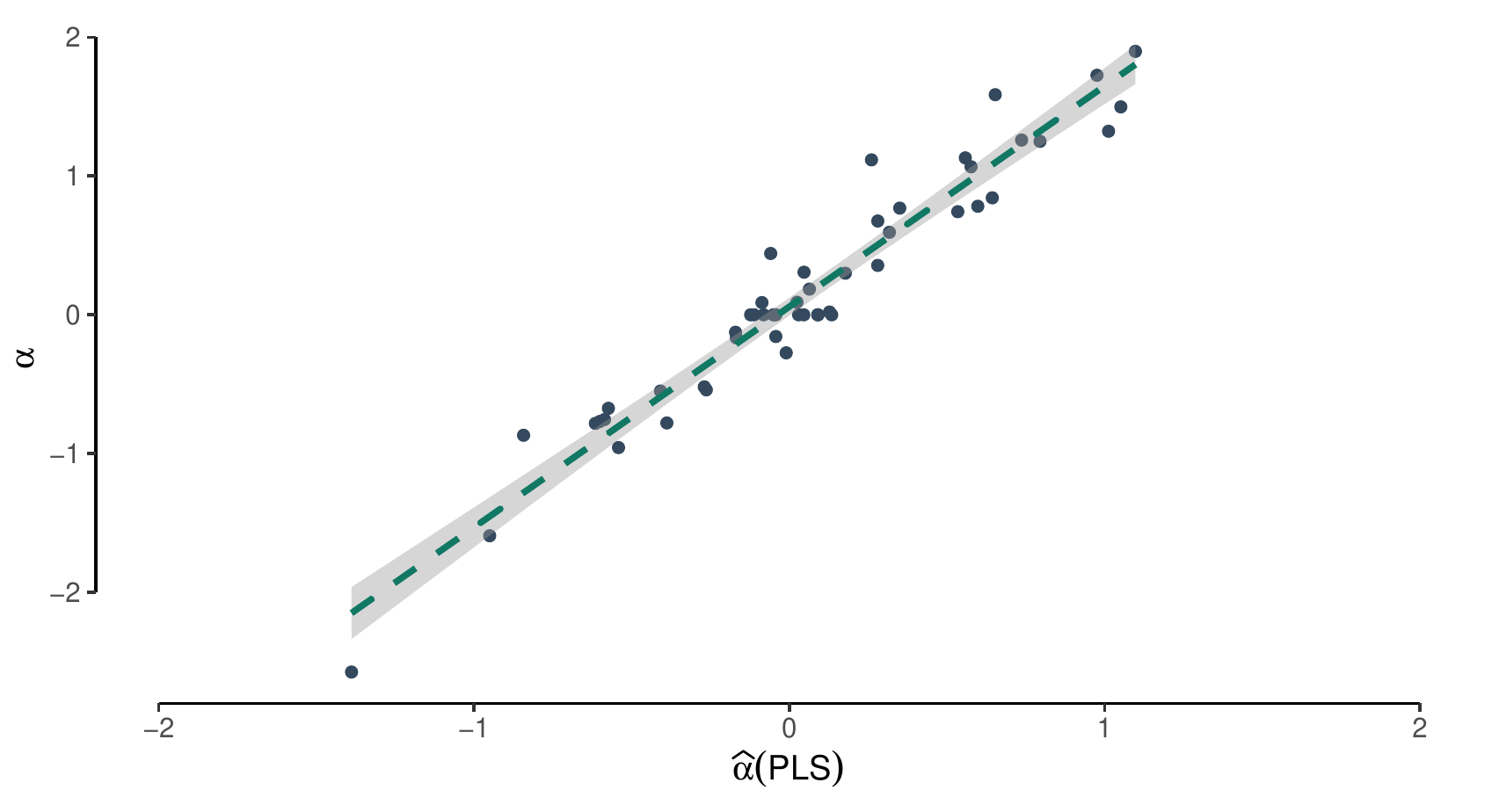}}\\	
	\caption{Estimated coefficients and their corresponding true values in the treatment network. The X-axis shows coefficients predicted by OLS (left) and PLS (right), and the Y-axis depicts the corresponding true values. }\label{fig_coeff_sim1}
\end{figure}

\begin{figure}[H]
	\centering
	\subfloat[]{\includegraphics[width=0.5\linewidth]{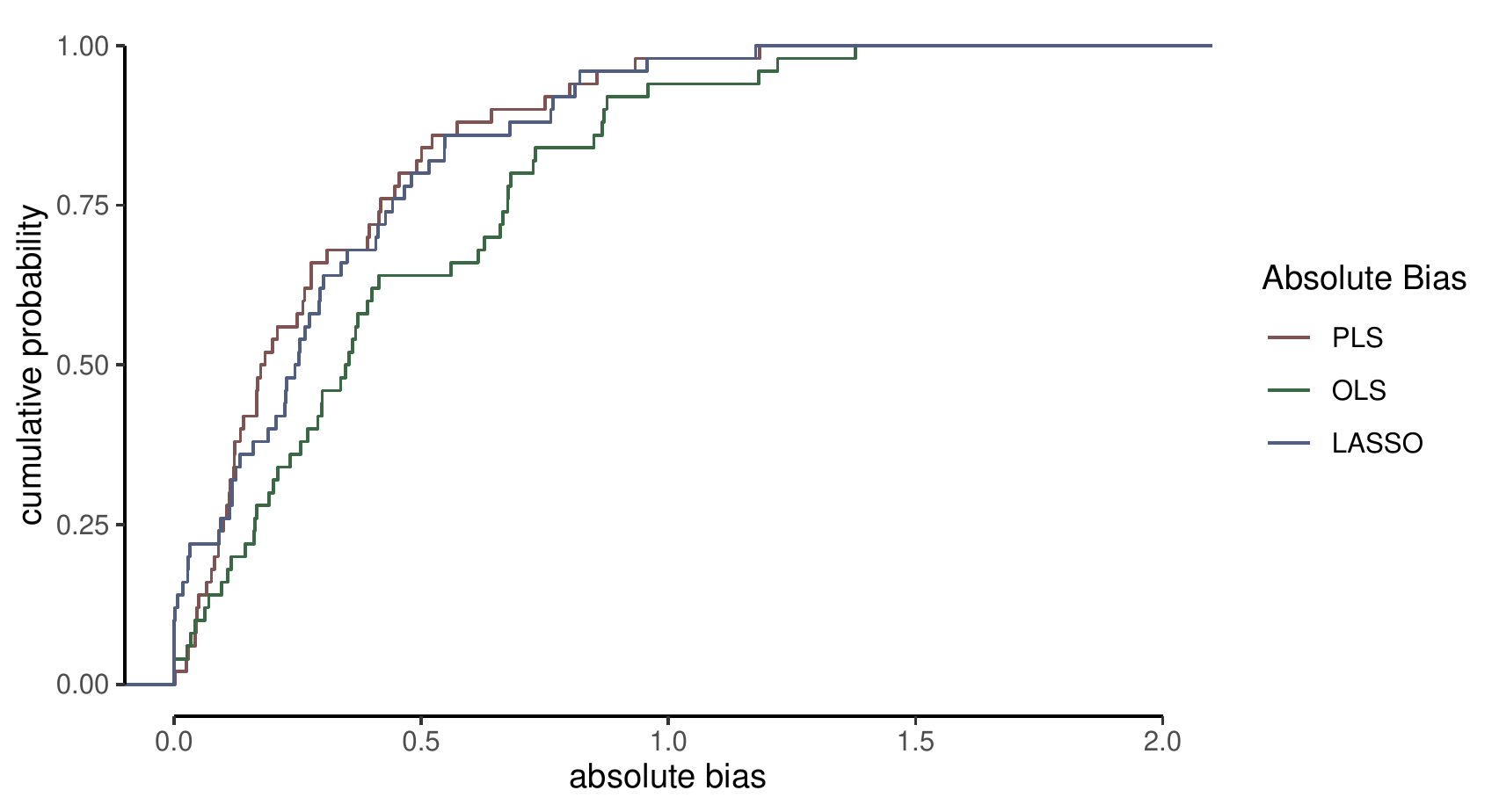}} \\
	\caption{Empirical cumulative distribution functions (CDF) of the absolute bias of the parameters estimated by PLS, OLS and LASSO. }\label{fig_cdf_sim1}
\end{figure}

To introduce uncertainty in the parameters of interest, we extend DPLS-IV to a Bayesian setup and compare it to a classical Bayesian IV approach. 
\begin{table}[H] \centering 
\begin{tabular}{@{\extracolsep{5pt}}lccccccc} 
\\[-1.8ex]\hline 
\hline \\[-1.8ex] 
Outcome network \\
\hline \\[-1.8ex] 
Measures & \multicolumn{1}{c}{Bayesian IV} & \multicolumn{1}{c}{Bayesian DPLS-IV}  \\ 
\hline \\[-1.8ex] 
$R^2$ & 0.778  & \textbf{0.838}  \\ 
RMSE & 4.375  & \bf{3.092}  \\ 
\hline \\[-1.8ex] 
\end{tabular}
\caption{Prediction performance of Bayesian DPLS-IV and IV methods. The measures are computed based on the mean prediction out of 10,000 predicted outcome variables.}\label{tab_bayesian} 
\end{table}
Table \ref{tab_bayesian} shows that Bayesian DPLS-IV (with ReLU as an activation function and two hidden layers) significantly outperforms its' linear counterpart. Figure \ref{fig_bayesian_sim1} verifies that the Bayesian DPLS-IV closely replicates the density of the original outcome variable. 

\begin{figure}[H]
	\centering
	\subfloat[]{\includegraphics[width=\linewidth]{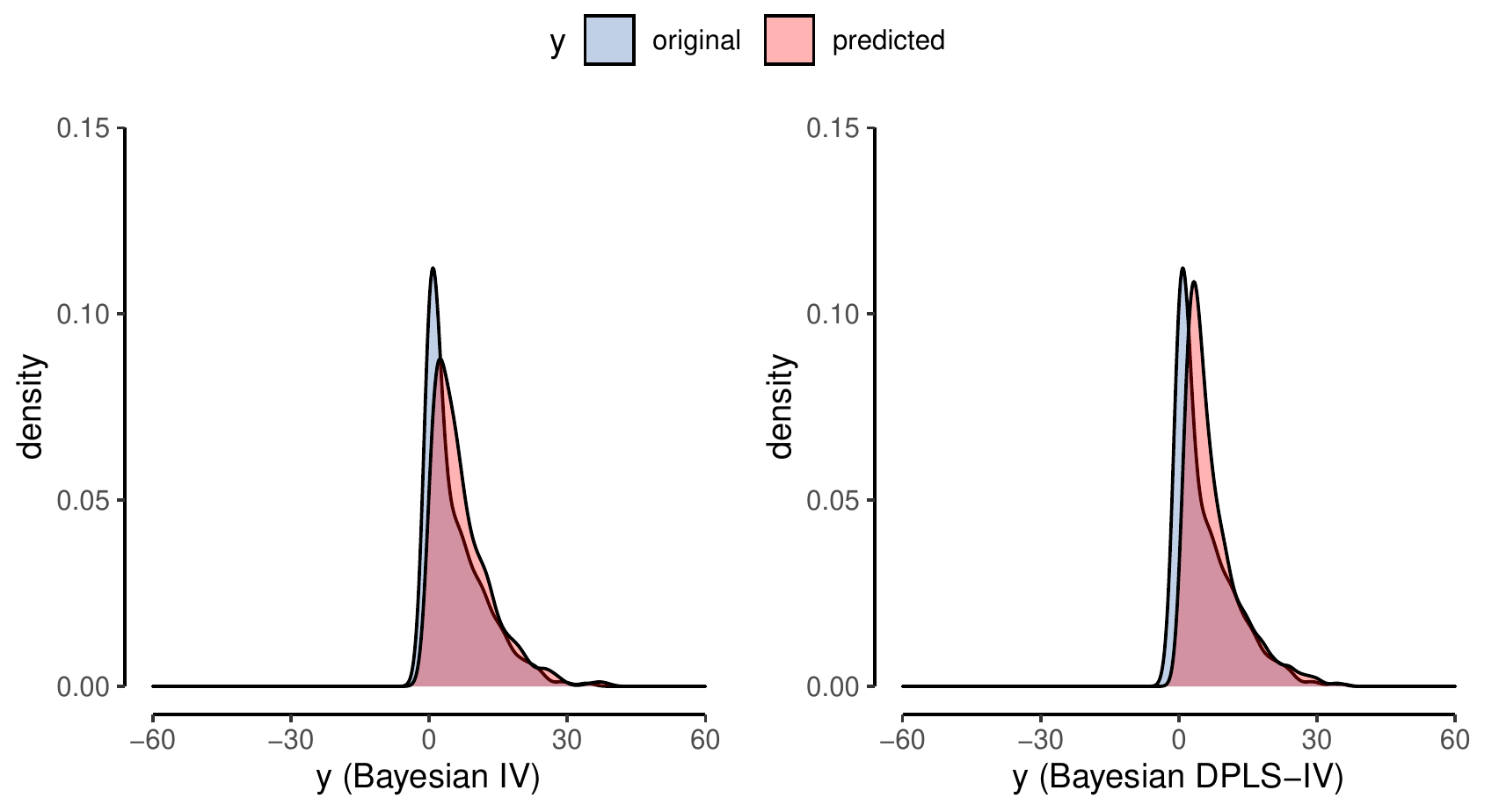}} \\
	\caption{The density of the original and predicted outcome variables (on the test data). }\label{fig_bayesian_sim1}
\end{figure}

\subsection{Instruments with the Network Structure}
The second simulation closely follows the network data structure described by \cite{kong2018deep}. The data-generating process is the same as shown in  \eqref{eq_simu}. However, the instrumental variable network comes from the preferential attachment algorithm (\citealp{jeong2003measuring}). Each node of the network represents one feature. The resulting network follows a power-law degree distribution, and thus, is scale-free. That means, only a few instruments in the network have a relatively large number of "neighbors". The distance between two instruments is the shortest path between them in the network. We calculate a $p \times p$ ($p = 50$) pairwise distance matrix $D$. Next, this distance matrix is transformed into a covariance matrix $\Sigma_{z, (i,j)} = 0.7^{D_{(i, j)}}$, where $(i, j)$ represents the element in each row $i$ and column $j$ of a matrix $D$ ($i, j = 1, \dots, p$). 
Figure \ref{fig_network_ex} shows an example of such a network with 100 nodes. 
\begin{figure}[H]
	\centering
	\subfloat[]{\includegraphics[width=0.8\linewidth]{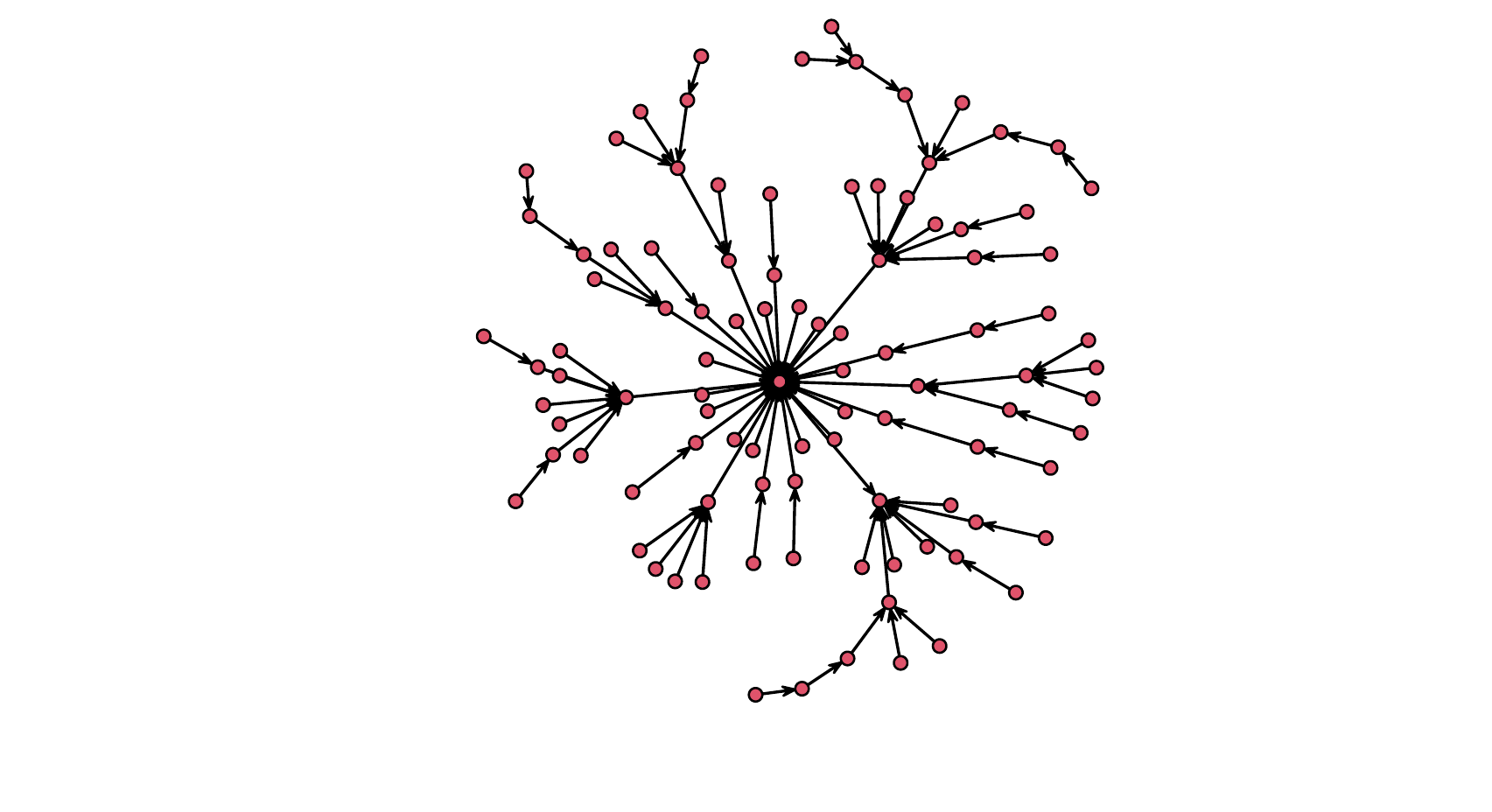}} \\
	\caption{A network with 100 nodes. }\label{fig_network_ex}
\end{figure}

\begin{figure}[H]
	\centering
	\subfloat[]{\includegraphics[width=0.95\linewidth]{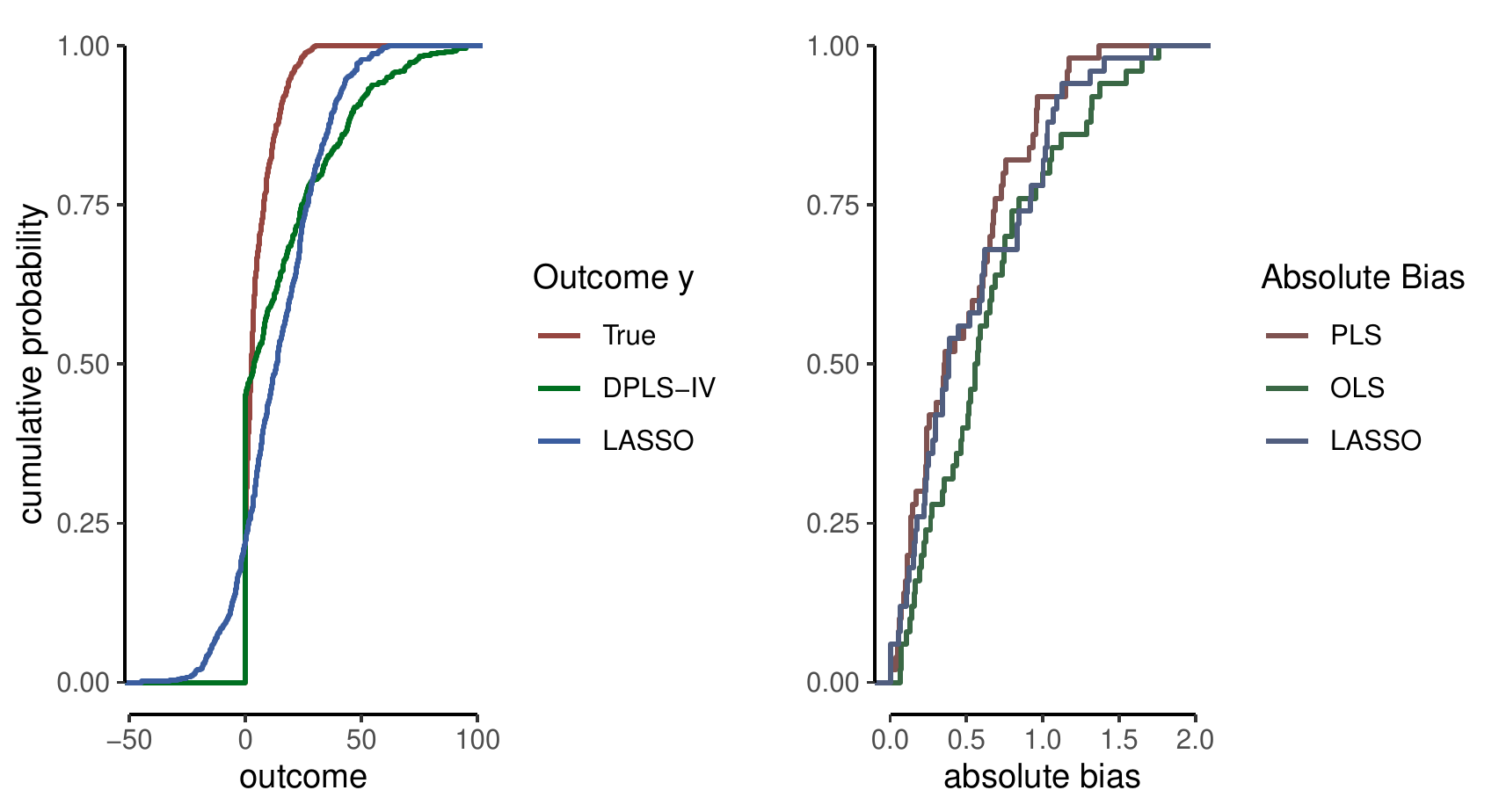}} \\
	\caption{Left - CDF of the outcome variable; Right - CDF of the absolute value of the bias of the parameters (estimated by the corresponding method) in the treatment network.}\label{fig_cdfs_network}
\end{figure}

For comparability, DPLS-IV and DeepIV consist of the same number of layers and neurons in each layer. In particular, the input layer consists of 200 neurons; the hidden layers consist of 200, 100, and 50 neurons, respectively. We use ReLU as an activation function. Figure \ref{fig_cdfs_network} shows cumulative distribution functions of the predicted outcome variable and the absolute bias of the estimated coefficients (right) in the treatment network. Based on the results in Figure \ref{fig_cdfs_network}, the outcome predicted by DPLS-IV is closer to the CDF of the simulated outcome variable. Moreover, DPLS-IV yields a CDF of the absolute bias of the parameters that stochastically dominates the CDF of the other benchmark methods. 

We also investigate prediction performance in the outcome network. Table \ref{tab_networks_pref} shows $R^2$ and RMSE of DPLS-IV relative to other methods. In this setting, DPLS-IV outperforms other benchmark methods. 

\begin{table}[H] \centering 
\begin{tabular}{@{\extracolsep{5pt}}lccccccc} 
\\[-1.8ex]\hline 
\hline \\[-1.8ex] 
Outcome network \\
\hline \\[-1.8ex] 
Measures & \multicolumn{1}{c}{PLS} & \multicolumn{1}{c}{OLS} &  \multicolumn{1}{c}{DeepIV} & \multicolumn{1}{c}{LASSO}  & \multicolumn{1}{c}{DPLS-IV} \\ 
\hline \\[-1.8ex] 
$R^2$ & 0.940 & 0.930 & 0.912 & 0.942 & \textbf{0.944} & \\ 
RMSE & 1.631 & 1.846  & 2.108  & 1.609 & \textbf{1.585} & \\ 
\hline \\[-1.8ex] 
\end{tabular}
\caption{Prediction performance of DPLS-IV relative to other methods. We present out-of-sample $R^2$ and RMSE. To present the maximum prediction performance of OLS, PLS, and LASSO in the outcome network, we use residuals predicted by DPLS-IV in the first stage.}\label{tab_networks_pref} 
\end{table}

\begin{figure}[H]
	\centering
	\subfloat[]{\includegraphics[width=0.7\linewidth]{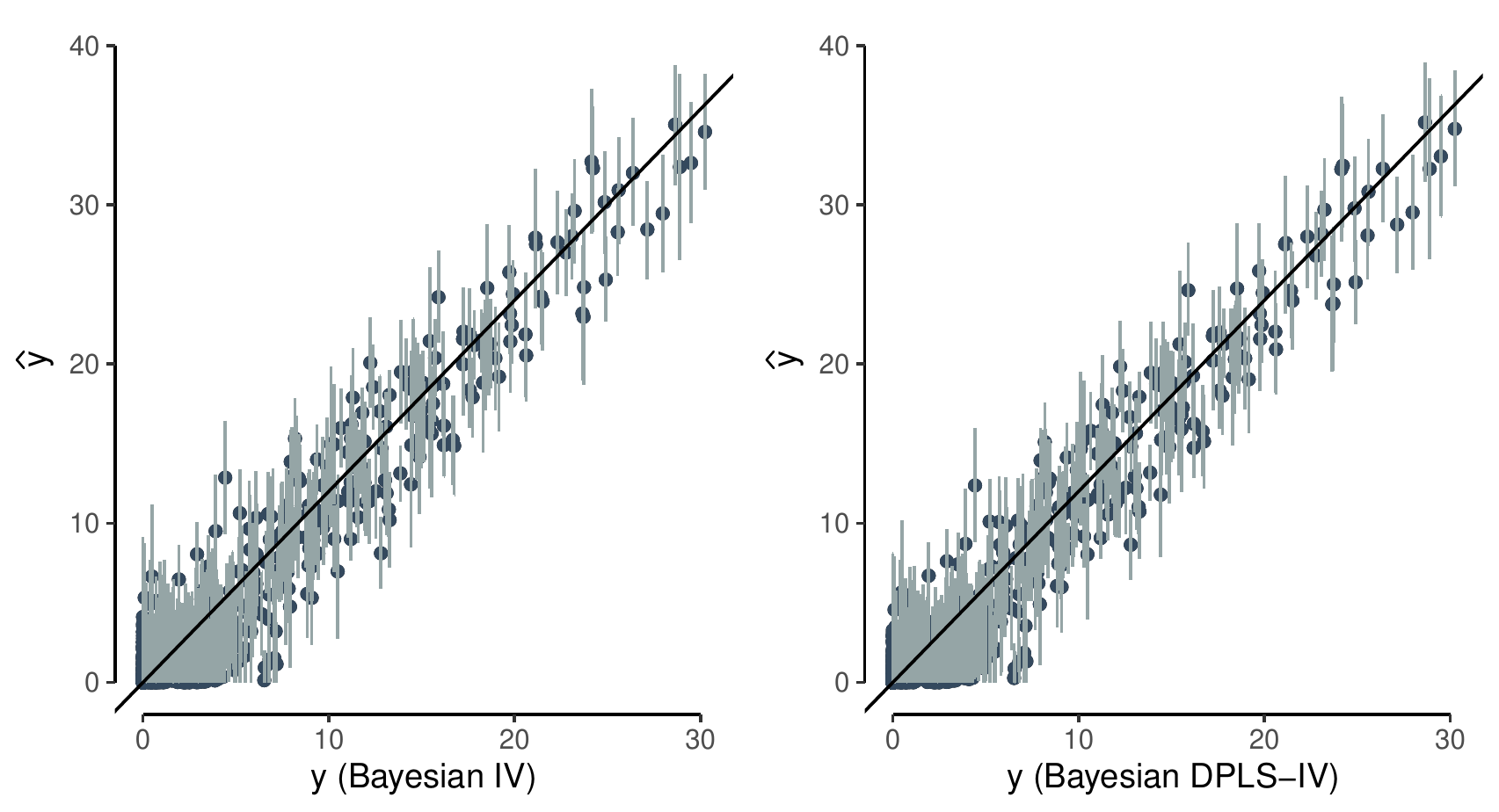}} \\
	\caption{Predicted outcome $y$ and the corresponding 95\% confidence intervals.}\label{fig_bayesian_network}
\end{figure}

Additionally, we compare a Bayesian DPLS-IV to a Bayesian IV approach. Bayesian DPLS-IV consists of two hidden layers and a ReLU activation function. Note that, the RMSE based on the mean predicted outcome out of 10,000 predictions increases (2.262) relative to DPLS-IV.  However, $R^2$ is unchanged (0.941). Figure \ref{fig_bayesian_network} shows predicted outcome values against their true counterparts. Bayesian DPLS-IV leads to outcome predictions that are closer to the true values.

\subsection{Labor Supply of Women}
\cite{angrist1996children} examine the effect of childbearing on women's labor supply. They use a mixed sibling-sex composition and twins as instruments for the size of the family. To illustrate the method, our analysis uses 1980 U.S. Census data that include all women with two or more children. The model that we are going to estimate is defined by the following structural equation model:
\begin{align}
   y &= f(\text{kids}\cdot\beta + x\beta_x + \xi) + \varepsilon. \\
      \text{kids} &= g(\text{twins}\cdot \alpha + \text{twins}\cdot x\gamma + x\alpha_x)  + w, 
  \nonumber
\end{align}
where $y$ is the outcome variable and measures the logarithm of hours worked per week by the mother. The outcome is observed when the age of the mother is more than the average age of the mothers in the population. In particular, 
\begin{align}
    y =  \text{log(hourswm)} \times \mathbb{I}(agem > \mathbb{E}(agem)),
\end{align}
where $\mathbb{I}(agem > \mathbb{E}(agem))$ is an indicator variable and equals one if the age of the mother is more than the population mean, and zero otherwise. The treatment, $\text{kids}$, is the number of total kids in a family. $\text{twins}$ represents an instrumental variable and equals one if the second and third children are twins, otherwise zero. Additionally, we use interactions of the instrument with covariates, $ \text{twins}\cdot x$ as instruments for the number of kids. The covariates include the gender and age of the first and second child, the mother's age, marital status, race, education, and the age of the mother when she first gave birth. See \cite{angrist1996children} for a detailed summary of the variables. 

The prediction performance advantage of DPLS-IV is also clearly evident in Table \ref{tab_angrist}. The first stage prediction performance of  DPLS-IV  is similar in each method, however, $R^2$ (RMSE) is considerably high (low) in the outcome network. It is worth noting that the measures become substantially worse in the outcome network compared to the ones in the treatment network. One of the potential reasons is that the outcome distribution is bimodal. Figure \ref{fig_bayesian_ae} presents the original density of the outcome variable and the density of the outcome predicted by Bayesian DPLS-IV.  Figure \ref{fig_bayesian_ae} shows that Bayesian  DPLS-IV  successfully captures the bimodal nature of the outcome.

\begin{table}[H] \centering 
\begin{tabular}{@{\extracolsep{5pt}}lccccccc} 
\\[-1.8ex]\hline 
\hline \\[-1.8ex] 
Treatment network \\
\hline \\[-1.8ex] 
Measures & \multicolumn{1}{c}{PLS} & \multicolumn{1}{c}{OLS} &  \multicolumn{1}{c}{DeepIV} &\multicolumn{1}{c}{LASSO} &   \multicolumn{1}{c}{DPLS-IV}\\ 
\hline \\[-1.8ex] 
$R^2$ & 0.204 & 0.205  & 0.226 & 0.205 & \textbf{0.233} \\ 
RMSE & 0.655 & 0.655  & \textbf{0.647} & \textbf{0.647} &  \textbf{0.644} \\
\hline \\[-1.8ex] 
Outcome network \\
\hline \\[-1.8ex] 
$R^2$ & 0.532 & 0.536  & \textbf{0.771} & 0.536 & \bf{0.772} \\ 
RMSE & 12.838 & 12.782  & 8.980 & 12.780 & \bf{8.960} \\ 
\hline \\[-1.8ex] 
\end{tabular}\caption{Prediction performance of DPLS-IV relative to other methods. We present out-of-sample $R^2$ and RMSE. To present the maximum prediction performance of OLS, PLS  and LASSO in the outcome network, we use residuals predicted by DPLS-IV in the first stage.}\label{tab_angrist} 
\end{table}
We also investigate the layer-by-layer transformation of the proposed method. Figure \ref{fig_tsne} illustrates the feature representation of the original outcome $y$ (on the test data),  the score matrix ($T$), and the final feed-forward neural network layer in DPLS-IV. We consider the projection of features on two distinct classes of the outcome variable. In particular, when $y > 0$, we label it as a class 1, and when $y = 0$, it represents a class 0.  Figure \ref{fig_tsne} shows that the intermediate layers in DPLS-IV significantly improve the representation of the original covariate space. The borders of the two classes become highly evident in the last layer of DPLS-IV. 

\begin{figure}[H]
	\centering
	\subfloat[Original]{\includegraphics[width=0.4\linewidth]{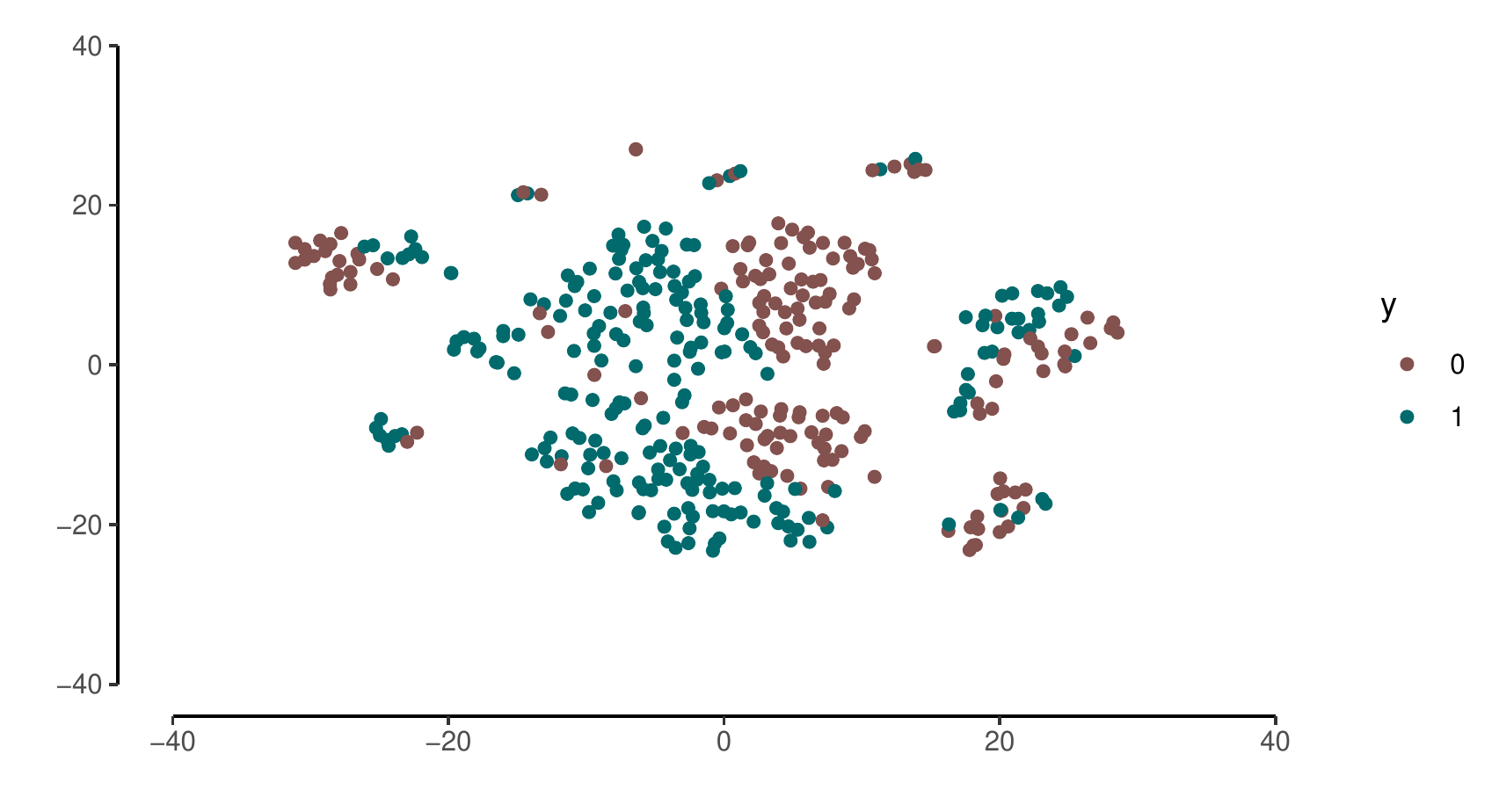}}\quad
	\subfloat[PLS Scores T]{\includegraphics[width=0.4\linewidth]{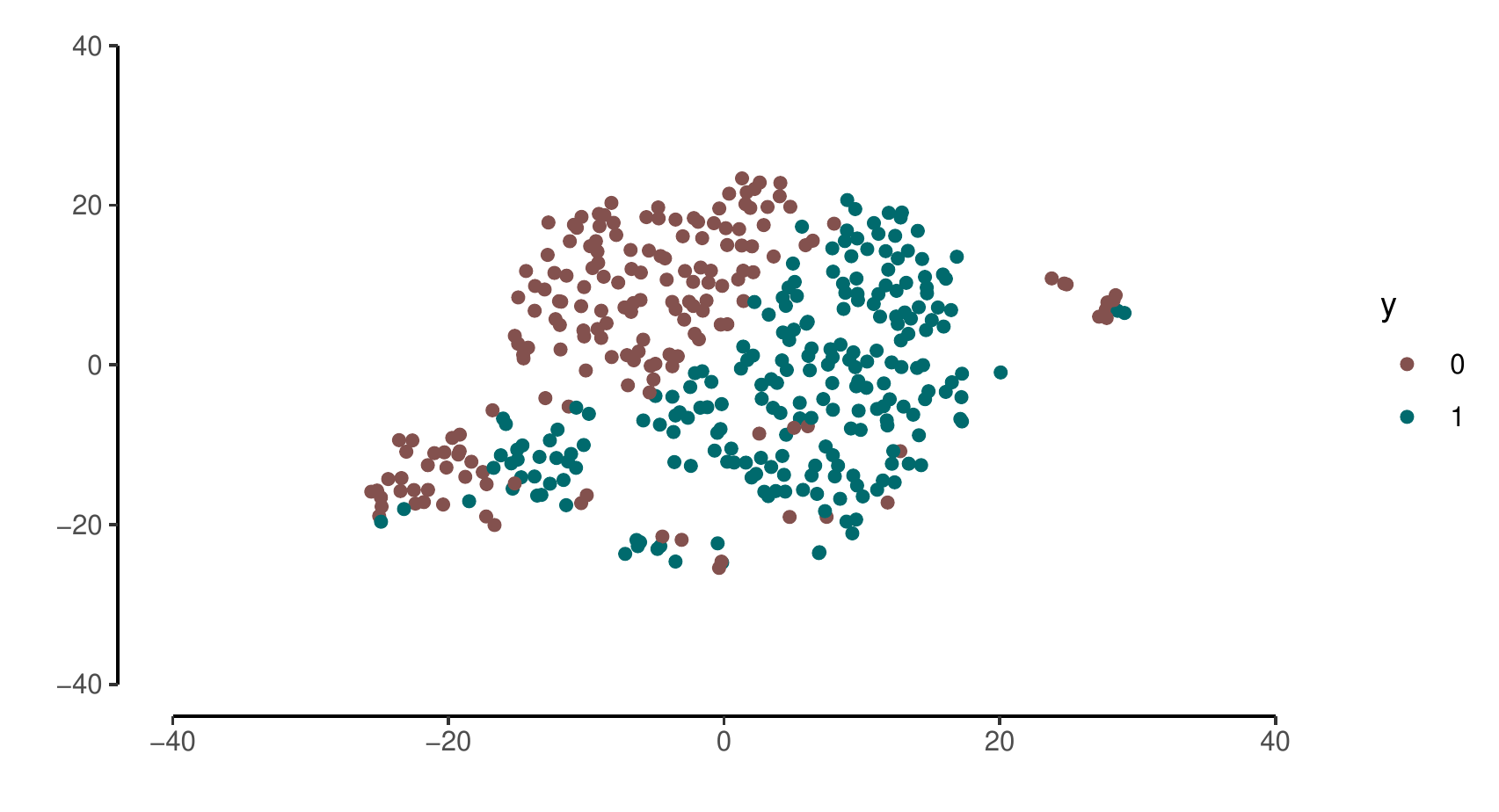}} \\
	\subfloat[Final Layer]{\includegraphics[width=0.4\linewidth]{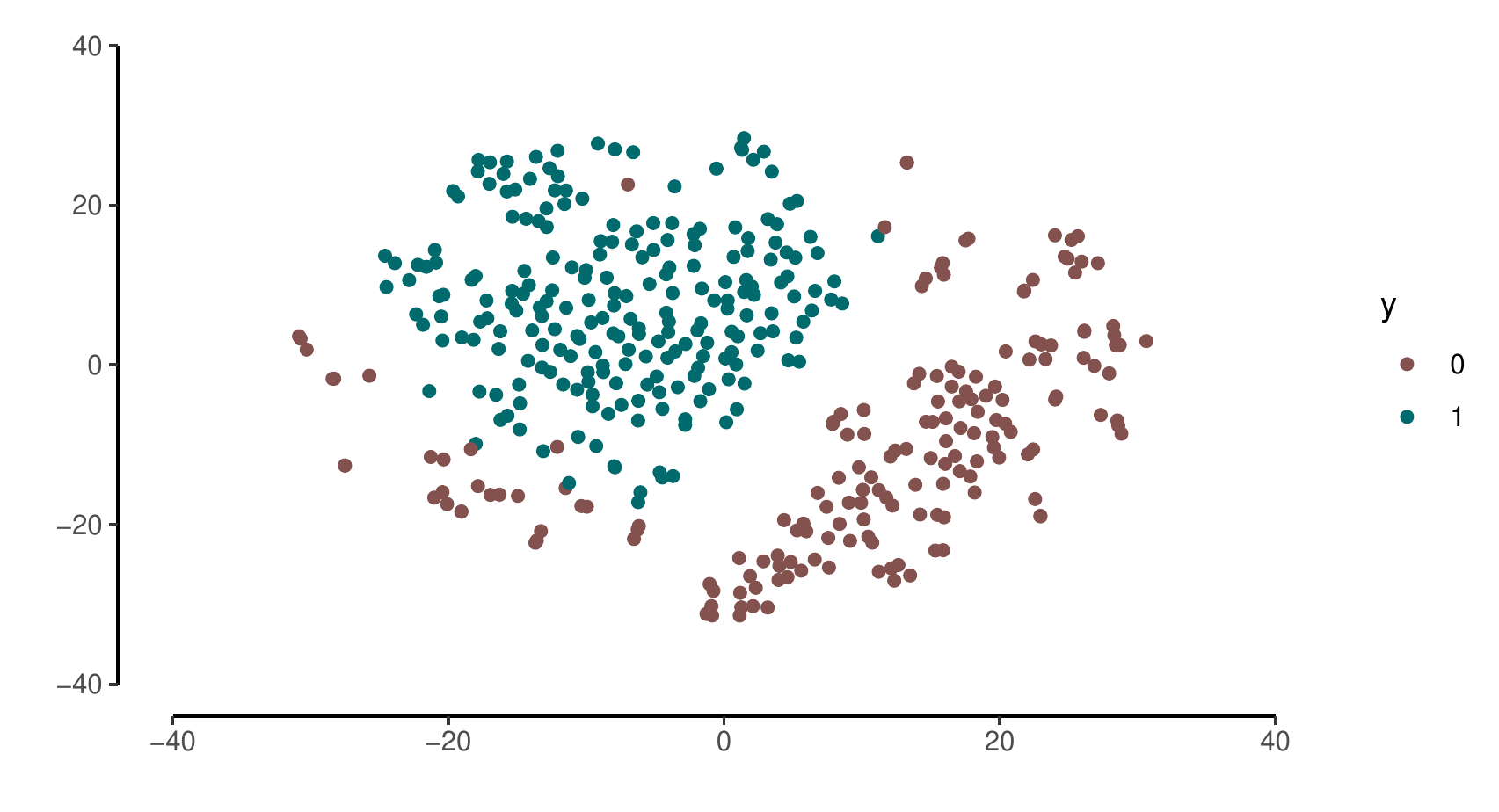}} \\
	\caption{Feature representation of DPLS-IV based on the t-SNE algorithm \citep{wattenberg2016use}, with the perplexity parameter equal to 50. Features are reduced to two distinct classes of the outcome $y$, with a class 1 when $y > 0$, and 0 otherwise. The final layer represents the predicted $y$ of individual neurons in the last layer of DPLS-IV. For computational benefits, we randomly sample 1000 observations from the original data. }\label{fig_tsne}
\end{figure}

\section{Discussion} \label{sect_disc}
In this article, we propose deep partial least squares for reducing the dimension of the instrumental variable space. The deep partial least squares method efficiently extracts features based on partial least squares and further processes the input with a feed-forward deep learner. The method is well-tailored for correlated instruments with sparse and nonlinear structures. More importantly, deep partial least squares are consistent, up to a proportionality constant. 

The applications on synthetic data as well as the application to the effect of childbearing on the mother's labor supply (\citealp{angrist1996children}) show that the deep partial least squares method outperforms other related methods. Moreover, a flexible number of layers allows us to efficiently capture nonlinearities embedded in the instrumental variable network.  

An interesting extension of this work is to consider a Bayesian model with various priors on the coefficients of interest. \cite{datta2013asymptotic} consider asymptotic properties of Bayes risk with the horseshoe prior. We believe, investigating different prior beliefs can result in the increased predictive performance of deep partial least squares. Another useful extension is to draw applications to eminent domain with judge characteristics as instruments (\citealp{belloni2012sparse}). In addition to prediction problems, the consistency of DPLS allows us to address the estimation of the treatment effect. In the future, we plan to demonstrate the precision of the treatment effect and compare it to other benchmark algorithms.

\bibliographystyle{apalike}
\bibliography{notes}

\section*{Appendix}
\subsection{Prediction Performance Measures with LeakyReLu}\label{app_leakyrelu}
\begin{table}[H] \centering 
\begin{tabular}{@{\extracolsep{5pt}}lccccccc} 
\\[-1.8ex]\hline 
\hline \\[-1.8ex] 
Treatment network \\
\hline \\[-1.8ex] 
Measures & \multicolumn{1}{c}{PLS} & \multicolumn{1}{c}{OLS} & \multicolumn{1}{c}{LASSO} &  \multicolumn{1}{c}{DeepIV} &  \multicolumn{1}{c}{DPLS-IV}\\ 
\hline \\[-1.8ex] 
$R^2$ & 0.757 & 0.696 & 0.760 & 0.939 & \textbf{0.957} \\ 
RMSE & 10.503 & 21.404 & 10.398 & 5.498 & \bf{4.449} \\ 
\hline \\[-1.8ex] 
Outcome network \\
\hline \\[-1.8ex] 
$R^2$ & 0.933 & 0.933 & 0.934 & 0.878 & \textbf{0.938} \\ 
RMSE & 1.666 & 1.720 & 1.667 & 2.224 & \bf{1.623} \\ 
\hline \\[-1.8ex] 
\end{tabular}\
\caption{Prediction performance of DPLS-IV relative to other methods. We present out-of-sample $R^2$ and RMSE. To present the maximum prediction performance of OLS, PLS and Tobit in the outcome network, we use residuals predicted by DPLS-IV in the first stage. We use LeakyReLU to simulate the outcome variables. DPLS-IV consists of an input layer with 50 neurons, and a hidden layer with 30 neurons. DeepIV consists of 100 neurons in the input layer, followed by three hidden layers with 100, 100 and 30 neurons, respectively. We use ReLU to activate neurons in DeepIV and DPLS-IV. }\label{tab_networks_lrelu} 
\end{table}

\begin{figure}[H]
	\centering
	\subfloat[OLS]{\includegraphics[width=0.4\linewidth]{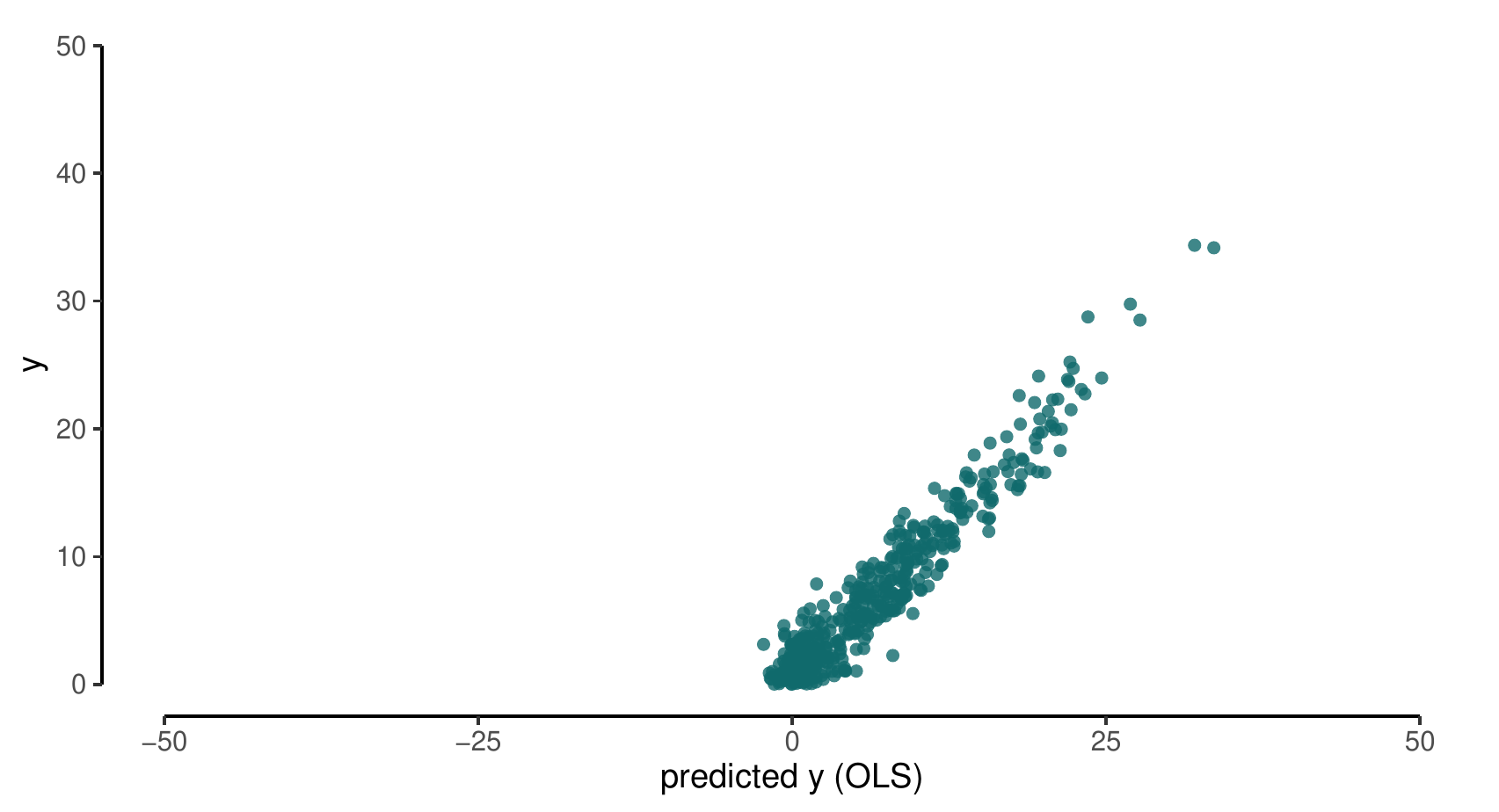}}\quad
	\subfloat[PLS]{\includegraphics[width=0.4\linewidth]{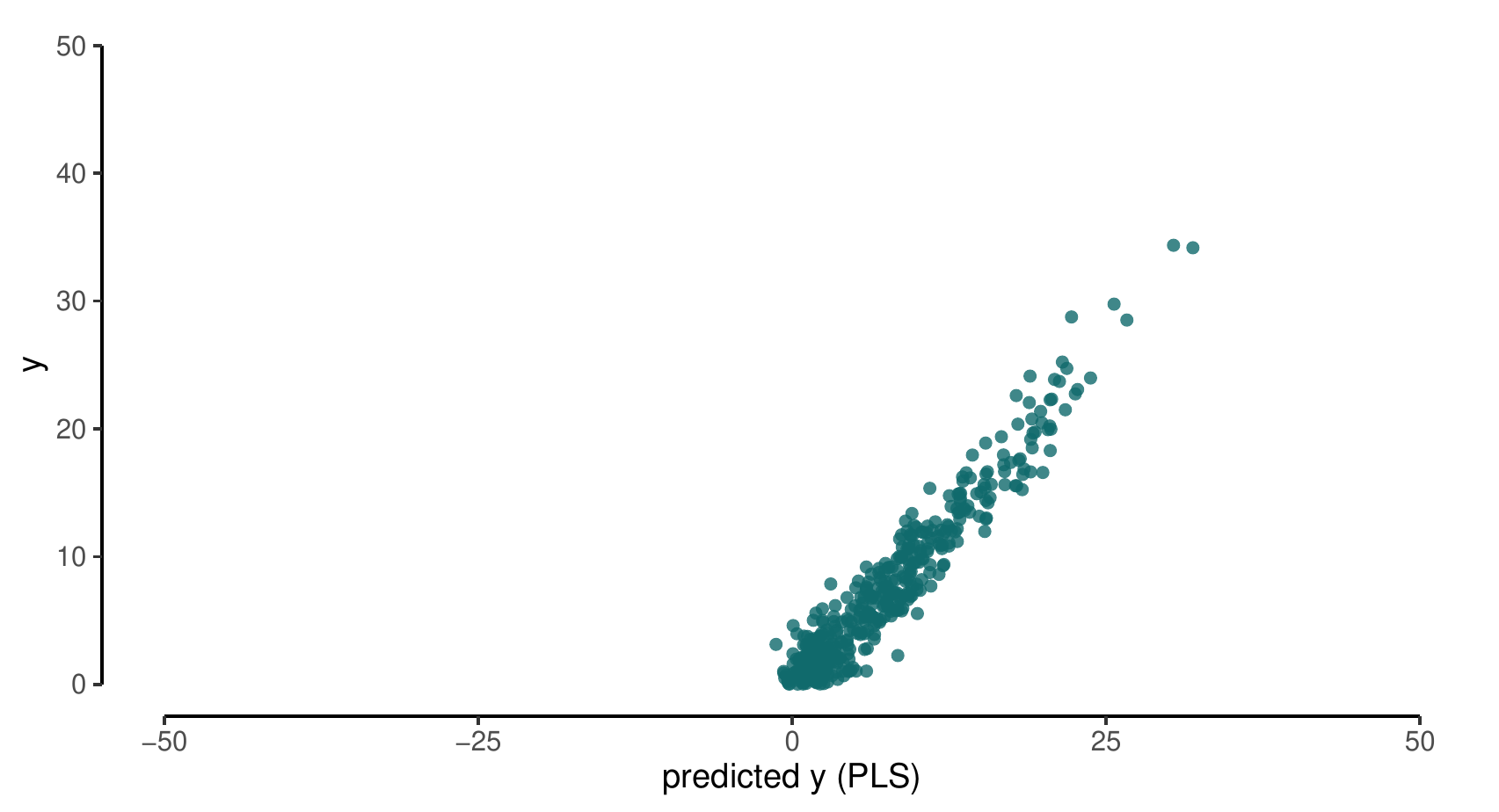}} \\
	\subfloat[DeepIV]{\includegraphics[width=0.4\linewidth]{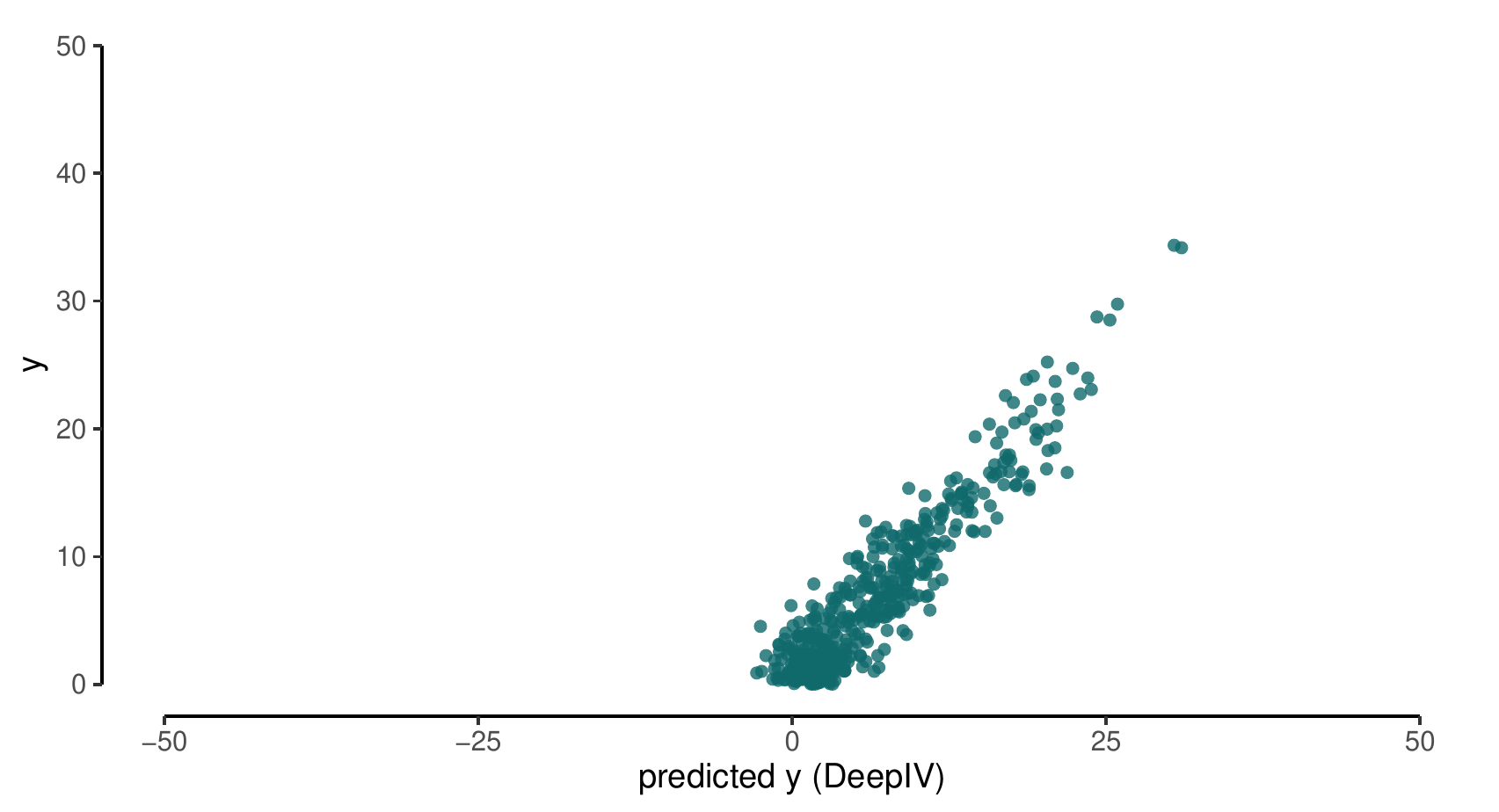}} 
	\subfloat[LASSO]{\includegraphics[width=0.4\linewidth]{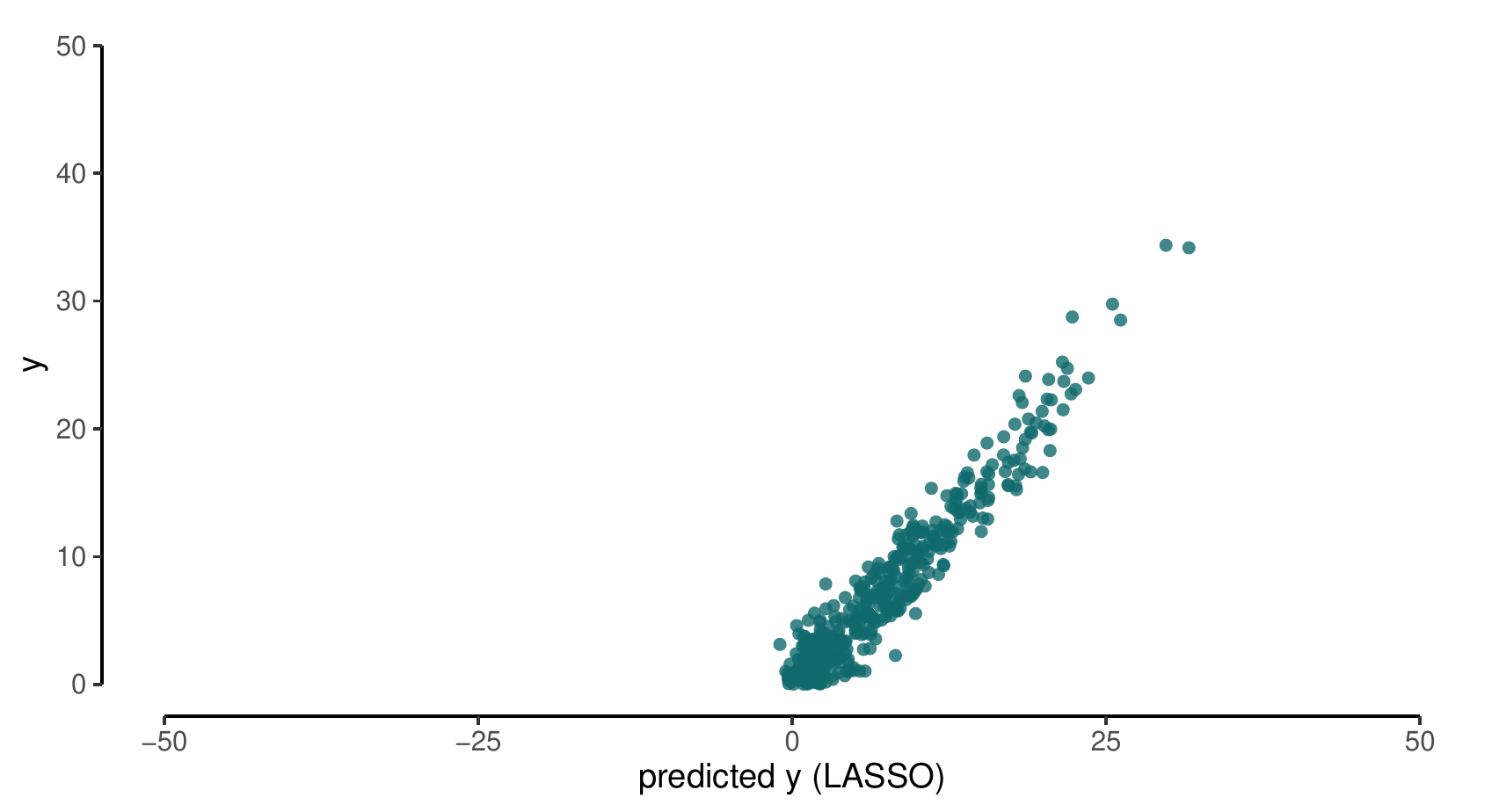}}\quad\\
	\subfloat[DPLS-IV]{\includegraphics[width=0.4\linewidth]{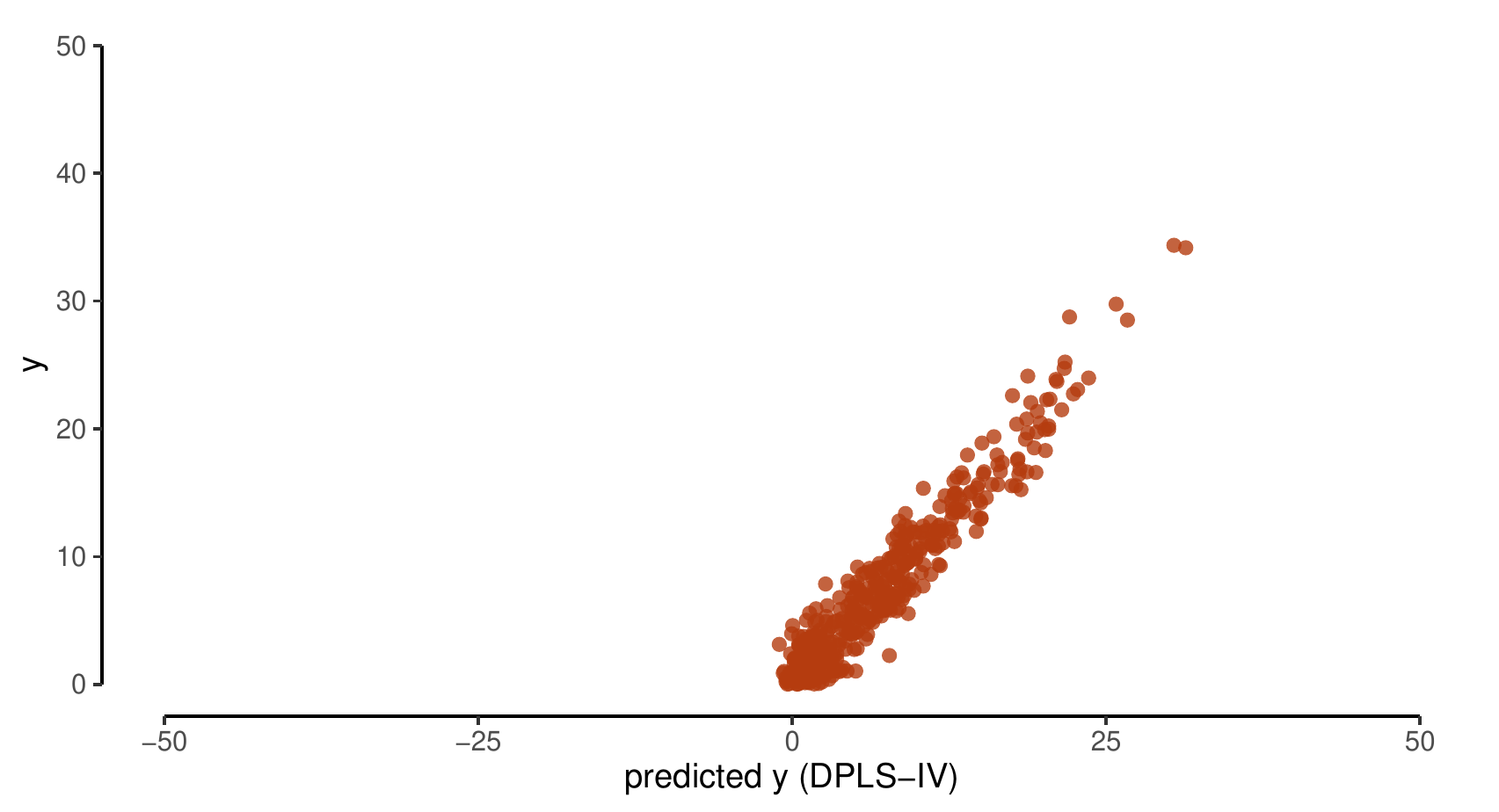}} \\
	\caption{The second stage prediction performance. The X-axis depicts predicted outcome $\widehat{y}$, and the Y-axis represents true values of $y$. DPLS-IV denotes the method introduced in this study. We use test data for evaluating the methods. }\label{fig_secondstage}
\end{figure}

\begin{figure}[H]
	\centering
	\subfloat[]{\includegraphics[width=0.7\linewidth]{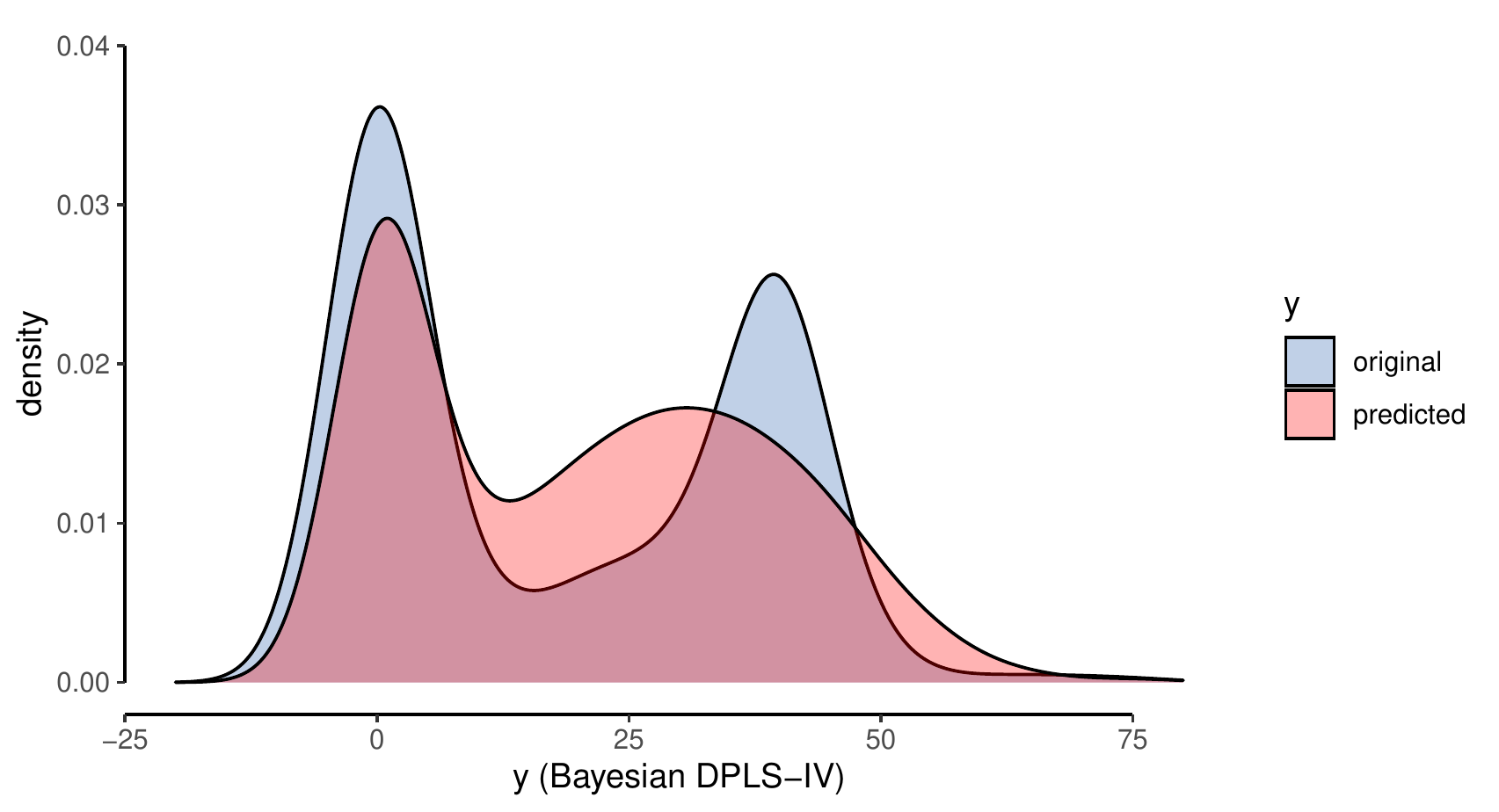}} \\
	\caption{The original distribution of the outcome variable and the one predicted by Bayesian DPLS-IV. Bayesian DPLS-IV consists of two hidden layers and a ReLU activation function. }\label{fig_bayesian_ae}
\end{figure}

\end{document}